\documentclass[sigconf]{acmart}

\usepackage{amsmath,amsfonts}
\usepackage{adjustbox}
\usepackage{algorithm}
\usepackage[noend]{algpseudocode}
\usepackage{algorithmicx}
\usepackage[toc,page]{appendix}
\usepackage{array}
\usepackage{booktabs}
\usepackage{collcell}
\usepackage{color,colortbl}
\usepackage{diagbox}
\usepackage{enumitem}
\usepackage[utf8]{inputenc}
\usepackage{multicol}
\usepackage{multirow}
\usepackage[final]{pdfpages}
\usepackage{pgfplots}	
\pgfplotsset{compat = 1.16}
\usepackage{placeins}
\usepackage{rotating}
\usepackage{tikz}
\usetikzlibrary{matrix,shapes,arrows,positioning,chains, calc, automata,arrows.meta}
\usepackage{todonotes}
\usepackage{tablefootnote}
\usepackage{makecell}

\newcommand\copyrightnotice[1]{
	\begin{tikzpicture}[remember picture,overlay]
		\node[anchor=south,yshift=10pt] at (current page.south) {\fbox{\parbox{\dimexpr\textwidth-\fboxsep-\fboxrule\relax}{#1}}};
	\end{tikzpicture}
}

\renewcommand{\arraystretch}{0.85}

\newcommand\Tstrut{\rule{0pt}{2.2ex}}         
\newcommand\Bstrut{\rule[-0.9ex]{0pt}{0pt}}   

\newtheorem{thm}{Theorem}
\newtheorem{definition}[thm]{Definition}

\newtheorem{functionality}{Functionality}

\makeatletter
\newcounter{protocol}

\makeatother

\makeatletter
\newenvironment{protocolDouble}[1][]{
	\let\c@algorithm\c@protocol
	\makeatletter
	\renewcommand{\ALG@name}{Protocol}
	\makeatother
	\begin{algorithm*}[#1]
	}{\end{algorithm*}
}
\makeatother

\makeatletter
\newcounter{gate}
\newenvironment{gate}[1][]{
	
	\let\c@algorithm\c@gate
	\renewcommand{\ALG@name}{Gate}
	
	\begin{algorithm}[#1]
	}{\end{algorithm}
}
\makeatother



\makeatletter
\newcounter{alg}

\newcommand{\ra}[1]{\renewcommand{\arraystretch}{#1}}

\setlist[enumerate]{label* = \arabic*.}
\setlistdepth{9}
\newlist{myEnumerate}{enumerate}{9}
\setlist[myEnumerate]{label* = \arabic*.}

\newcolumntype{L}[1]{>{\raggedright\let\newline\\\arraybackslash\hspace{0pt}}m{#1}}
\newcolumntype{C}[1]{>{\centering\let\newline\\\arraybackslash\hspace{0pt}}m{#1}}
\newcolumntype{R}[1]{>{\raggedleft\let\newline\\\arraybackslash\hspace{0pt}}m{#1}}

\tikzset{every tree node/.style={align=center}}
\definecolor{smartgray}{RGB}{61, 62, 64}
\definecolor{smartblue}{HTML}{0B6296}
\definecolor{smartred}{HTML}{C54949}
\definecolor{smartyellow}{HTML}{d0d62a}
\definecolor{smartgreen}{HTML}{61c15d}
\definecolor{smartorange}{HTML}{ff7e00}
\definecolor{smartviolet}{HTML}{d19fe8}
\definecolor{smartbrown}{HTML}{480607}
\definecolor{smartcadet}{HTML}{5F9EA0}

\pgfplotscreateplotcyclelist{color_comm}{%
	style={solid, color=smartred}, every mark/.append style={solid, fill=smartred}, mark=*\\%
	style={solid, color=smartred}, every mark/.append style={solid, fill=smartred}, mark=square*\\%
	style={solid, color=smartblue}, every mark/.append style={solid, fill=smartblue}, mark=*\\%
	style={solid, color=smartblue}, every mark/.append style={solid, fill=smartblue}, mark=square*\\%
	style={solid, color=green}, every mark/.append style={solid, fill=green}, mark=*\\%
	style={solid, color=black}, every mark/.append style={solid, fill=black}, mark=square*\\%
}

\pgfplotscreateplotcyclelist{color_default}{%
	style={solid, color=smartred}, every mark/.append style={solid, fill=smartred}, mark=*\\%
	style={solid, color=smartblue}, every mark/.append style={solid, fill=smartblue}, mark=square*\\%
	style={solid, color=smartgreen}, every mark/.append style={solid, fill=smartgreen}, mark=triangle*\\%
	style={solid, color=smartorange}, every mark/.append style={solid, fill=smartorange}, mark=*\\%
	style={solid, color=smartcadet}, every mark/.append style={solid, fill=smartcadet}, mark=*\\%
}

\pgfplotscreateplotcyclelist{color_boxplot_data}{%
	style={solid, color=smartgreen}, every mark/.append style={solid, fill=smartred}, mark=triangle*\\%
	style={solid, color=smartorange}, every mark/.append style={solid, fill=smartorange}, mark=square*\\%
	style={solid, color=smartred}, every mark/.append style={solid, fill=smartgreen}, mark=*\\%
	style={solid, color=smartblue}, every mark/.append style={solid, fill=smartblue}, mark=diamond*\\%
	style={solid, color=smartcadet}, every mark/.append style={solid, fill=smartcadet}, mark=star\\%
}

\pgfplotscreateplotcyclelist{color_with_extrapolation}{%
	style={solid, color=smartred}, every mark/.append style={solid, fill=smartred}, mark=*\\%
	style={solid, color=smartblue}, every mark/.append style={solid, fill=smartblue}, mark=*\\%
	style={solid, color=smartgreen}, every mark/.append style={solid, fill=smartgreen}, mark=*\\%
	style={dashed, color=smartred}, every mark/.append style={solid, fill=smartred}, mark=none\\%
	style={dashed, color=smartblue}, every mark/.append style={solid, fill=smartblue}, mark=none\\%
	style={dashed, color=smartgreen}, every mark/.append style={solid, fill=smartgreen}, mark=none\\%
}

\pgfplotscreateplotcyclelist{color_comp_ap_ip}{%
	style={solid, color=smartred}, every mark/.append style={solid, fill=smartred}, mark=*\\%
	style={solid, color=smartgreen}, every mark/.append style={solid, fill=smartgreen}, mark=square*\\%
	style={solid, color=smartred}, every mark/.append style={solid, fill=smartred}, mark=triangle*\\%
	style={solid, color=smartgreen}, every mark/.append style={solid, fill=smartgreen}, mark=diamond*\\%
	style={solid, color=smartblue}, every mark/.append style={solid, fill=smartblue}, mark=star\\%
}

\pgfplotscreateplotcyclelist{color_comp_ap_ip_crossover}{%
	style={solid, color=smartred}, every mark/.append style={solid, fill=smartred}, mark=*\\%
	style={solid, color=smartred}, every mark/.append style={solid, fill=smartred}, mark=triangle*\\%
	style={solid, color=smartgreen}, every mark/.append style={solid, fill=smartgreen}, mark=square*\\%
	style={solid, color=smartgreen}, every mark/.append style={solid, fill=smartgreen}, mark=diamond*\\%
}

\pgfplotscreateplotcyclelist{color_compared_four}{%
	style={solid, color=smartred}, every mark/.append style={solid, fill=smartred}, mark=*\\%
	style={solid, color=smartred}, every mark/.append style={solid, fill=smartred}, mark=triangle*\\%
	style={solid, color=smartblue}, every mark/.append style={solid, fill=smartblue}, mark=square*\\%
	style={solid, color=smartgreen}, every mark/.append style={solid, fill=smartgreen}, mark=diamond*\\%
}

\pgfplotscreateplotcyclelist{color_compared_three}{%
	style={solid, color=smartred}, every mark/.append style={solid, fill=smartred}, mark=*\\%
	style={solid, color=smartred}, every mark/.append style={solid, fill=smartred}, mark=square*\\%
	style={solid, color=smartgreen}, every mark/.append style={solid, fill=smartgreen}, mark=diamond*\\%
}

\pgfplotscreateplotcyclelist{color_with_two_networks}{%
	style={solid, color=smartred}, every mark/.append style={solid, fill=smartred}, mark=*\\%
	style={solid, color=smartblue}, every mark/.append style={solid, fill=smartblue}, mark=*\\%
	style={solid, color=smartgreen}, every mark/.append style={solid, fill=smartgreen}, mark=*\\%
	style={solid, color=smartred}, every mark/.append style={solid, fill=smartred}, mark=triangle*\\%
	style={solid, color=smartblue}, every mark/.append style={solid, fill=smartblue}, mark=triangle*\\%
	style={solid, color=smartgreen}, every mark/.append style={solid, fill=smartgreen}, mark=triangle*\\%
}

\pgfplotscreateplotcyclelist{color_with_two_bandwidths_three_latencies}{%
	style={solid, color=smartred}, every mark/.append style={solid, fill=smartred}, mark=triangle*\\%
	style={solid, color=smartred}, every mark/.append style={solid, fill=smartred}, mark=*\\%
	style={solid, color=smartblue}, every mark/.append style={solid, fill=smartblue}, mark=triangle*\\%
	style={solid, color=smartblue}, every mark/.append style={solid, fill=smartblue}, mark=*\\%
	style={solid, color=smartgreen}, every mark/.append style={solid, fill=smartgreen}, mark=triangle*\\%
	style={solid, color=smartgreen}, every mark/.append style={solid, fill=smartgreen}, mark=*\\%
}

\pgfplotscreateplotcyclelist{color_cmp_1}{%
	style={solid, color=smartred}, every mark/.append style={solid, fill=smartred}, mark=*\\%
	style={solid, color=smartblue}, every mark/.append style={solid, fill=smartblue}, mark=square*\\%
}

\pgfplotscreateplotcyclelist{color_cmp_2}{%
	style={solid, color=smartred}, every mark/.append style={solid, fill=smartred}, mark=*\\%
	style={solid, color=smartyellow}, every mark/.append style={solid, fill=smartyellow}, mark=*\\%
}

\pgfplotscreateplotcyclelist{color_cmp_3}{%
	style={solid, color=smartred}, every mark/.append style={solid, fill=smartred}, mark=*\\%
	style={solid, color=smartgreen}, every mark/.append style={solid, fill=smartgreen}, mark=*\\%
}

\pgfplotscreateplotcyclelist{color_greedy}{%
	style={solid, color=smartred}, every mark/.append style={solid, fill=smartred}, mark=*\\%
	style={solid, color=smartorange}, every mark/.append style={solid, fill=smartorange}, mark=*\\%
	style={solid, color=smartyellow}, every mark/.append style={solid, fill=smartyellow}, mark=*\\%
	style={solid, color=smartgreen}, every mark/.append style={solid, fill=smartgreen}, mark=*\\%
	style={solid, color=smartblue}, every mark/.append style={solid, fill=smartblue}, mark=*\\%
	style={solid, color=smartviolet}, every mark/.append style={solid, fill=smartviolet}, mark=*\\%
	style={solid, color=smartcadet}, every mark/.append style={solid, fill=smartcadet}, mark=*\\%
}

\pgfplotscreateplotcyclelist{color_comp}{%
	style={solid, color=smartred}, every mark/.append style={solid, fill=smartred}, mark=*\\%
	style={solid, color=smartblue}, every mark/.append style={solid, fill=smartblue}, mark=*\\%
	style={solid, color=smartviolet}, every mark/.append style={solid, fill=smartviolet}, mark=*\\%
	style={solid, color=smartgreen}, every mark/.append style={solid, fill=smartgreen}, mark=*\\%
}

\definecolor{RWTHblau}{RGB}{0,84,159}
\definecolor{RWTHmagenta}{RGB}{227,0,102}
\definecolor{RWTHgelb}{RGB}{255,237,0}
\definecolor{RWTHpetrol}{RGB}{0,79,101}
\definecolor{RWTHturkis}{RGB}{0,152,161}
\definecolor{RWTHgrun}{RGB}{87,171,39}
\definecolor{RWTHmaigrun}{RGB}{189,205,0}
\definecolor{RWTHorange}{RGB}{246,168,0}
\definecolor{RWTHrot}{RGB}{204,7,30}
\definecolor{RWTHbordeaux}{RGB}{161,16,53}
\definecolor{RWTHviolett}{RGB}{97,33,88}
\definecolor{RWTHlila}{RGB}{122,111,172}

\pgfplotscreateplotcyclelist{colorcyclefilled}{%
	solid, RWTHviolett, fill=RWTHviolett!60\\%
	solid, RWTHblau, fill=RWTHblau!60\\%
	solid, RWTHrot, fill=RWTHrot!60\\%
	solid, RWTHorange, fill=RWTHorange!60\\%
	solid, RWTHgrun, fill=RWTHgrun!60\\%
	solid, smartbrown, fill=smartbrown!60\\%
}

\makeatletter
\newcommand{\thickhline}{%
	\noalign {\ifnum 0=`}\fi \hrule height 1pt
	\futurelet \reserved@a \@xhline
}
\makeatother

\newcommand*{\MinNumber}{65}%
\newcommand*{\MidNumber}{90} %
\newcommand*{\MaxNumber}{100}%
\newcommand*{\celloffset}{0.2}

\newcommand{\ApplyGradient}[1]{%
	\ifdim #1 pt > \MidNumber pt
		\ifdim #1 pt < 100 pt
			\pgfmathsetmacro{\PercentColor}{max(min(100.0*(#1 - \MidNumber)/(\MaxNumber-\MidNumber),100.0),0.00)} %
			\hspace{-0.33em}\colorbox{smartgreen!\PercentColor!smartyellow}{\hspace*{0.543ex+\celloffset ex}#1\hspace*{0.543ex+\celloffset ex}}
		\else
			\pgfmathsetmacro{\PercentColor}{max(min(100.0*(#1 - \MidNumber)/(\MaxNumber-\MidNumber),100.0),0.00)} %
			\hspace{-0.33em}\colorbox{smartgreen!\PercentColor!smartyellow}{\hspace*{\celloffset ex}#1\hspace*{\celloffset ex}}
		\fi
	\else
		\pgfmathsetmacro{\PercentColor}{max(min(100.0*(\MidNumber - #1)/(\MidNumber-\MinNumber),100.0),0.00)} %
		\hspace{-0.33em}\colorbox{smartred!\PercentColor!smartyellow}{\hspace*{0.543ex+\celloffset ex}#1\hspace*{0.543ex+\celloffset ex}}
	\fi
}

\newcolumntype{R}{>{\collectcell\ApplyGradient}c<{\endcollectcell}}

\tikzset{MyArrow/.style={single arrow, draw=black, minimum width=10mm, minimum height=32mm, inner sep=0mm, single arrow head extend=1mm}}
\DeclareMathAlphabet{\pazocal}{OMS}{zplm}{m}{n}

\newcommand*{\proofCompCyclesSingle}{\ensuremath{c}}
\newcommand*{\proofOptCyclesSingle}{\ensuremath{o}}
\newcommand*{\proofCompCycles}{\ensuremath{C}}
\newcommand*{\proofOptCycles}{\ensuremath{O}}

\newcommand*{\compGraphMatrix}{\ensuremath{M}}
\newcommand*{\compGraph}{\ensuremath{G^C}}

\newcommand*{\bloodtypeIndicatorVector}[2]{\ensuremath{B^{#2}_{#1}}}
\newcommand*{\cg}{\ensuremath{G^\text{C}}}

\newcommand*{\compCheck}{\ensuremath{\gat_{\small{\textsc{comp-check}}}}}

\newcommand*{\prioGate}{\ensuremath{\gat_{\small{\textsc{prio}}}}}
\newcommand*{\prioMatrix}{\ensuremath{W}}

\newcommand*{\weightFunction}{\ensuremath{\omega}}
\newcommand*{\cycleBound}{\ensuremath{\kappa}}

\newcommand*{\weight}{\ensuremath{w}}

\newcommand*{\kepApProtocol}{\ensuremath{\prot_{\small{\textsc{kep-ap}}}}}
\newcommand*{\kepApFunc}{\ensuremath{\func_{\small{\textsc{kep-ap}}}}}
\newcommand*{\searchHeaviestCycle}{\ensuremath{\gat_{\small{\textsc{max-weight-set}}}}}

\newcommand*{\indicatorEmpty}{\ensuremath{\textit{ind}}}
\newcommand*{\demuxGate}[1]{\ensuremath{\gat_{\textsc{demux}^{#1}}}}
\newcommand*{\demuxVecGate}{\ensuremath{\gat_{\textsc{demux-vec}}}}

\newcommand*{\cycleWeight}{\ensuremath{\textit{weights}}}

\newcommand*{\kepApNodes}{\ensuremath{\textit{nodes}}}
\newcommand*{\heaviest}{\ensuremath{\textit{subset}}}
\newcommand*{\indicatorCombined}{\ensuremath{\textit{comb}}}

\newcommand*{\indexHeaviest}{\ensuremath{\textit{index}}}
\newcommand*{\chosenSets}{\ensuremath{\textit{chosen\_subsets}}}
\newcommand*{\chosenCycles}{\ensuremath{\textit{chosen\_cycles}}}
\newcommand*{\cycleIndices}{\ensuremath{\textit{indices}}}
\newcommand*{\firstCycle}{\ensuremath{\textit{first}}}
\newcommand*{\secondCycle}{\ensuremath{\textit{second}}}
\newcommand*{\chooseFirst}{\ensuremath{\textit{choose\_first}}}
\newcommand*{\subsetCycleMap}{\ensuremath{\textit{subset\_map}}}

\newcommand*{\shufflePermutation}{\ensuremath{\sigma}}

\newcommand*{\kepIpProtocol}{\ensuremath{\prot_{\small{\textsc{kep-ip}}}}}

\newcommand*{\solutionMatrix}{\ensuremath{A}}
\newcommand*{\outputDonors}{\ensuremath{D}}
\newcommand*{\outputRecipients}{\ensuremath{R}}
\newcommand*{\singleCycle}{\ensuremath{c}}

\newcommand*{\singleSubset}{\ensuremath{s}}
\newcommand*{\subsets}{\ensuremath{\mathcal{S}}}

\newcommand*{\inputQuote}[1]{\ensuremath{\mathcal{Q}_{#1}}}

\newcommand*{\sharedGreaterThan}{\ensuremath{\stackrel{?}{>}}}

\newcommand*{\sharedGreaterThanEqual}{\ensuremath{\stackrel{?}{\geq}}}
\newcommand*{\sharedMult}{\ensuremath{ \ \cdot \ }}

\newcommand*{\sharedScalMult}{\ensuremath{\times}}

\newcommand*{\shuffleGate}{\ensuremath{\gat_{\small{\textsc{shuffle}}}}}
\newcommand*{\reverseShuffleGate}{\ensuremath{\gat_{\small{\textsc{rev-shuffle}}}}}

\newcommand*{\bitsize}{\ensuremath{k}}
\newcommand*{\sharedDotProduct}{\ensuremath{\gat_{\small{\textsc{dot-product}}}}}
	
\newcommand*{\antigenvec}[1]{\ensuremath{A^{d}_{#1}}}

\newcommand*{\antibodyvec}[1]{\ensuremath{A^{p}_{#1}}}

\newcommand*{\perm}{\ensuremath{\sigma}}

\newcommand*{\ints}{\ensuremath{\mathbb{Z}}}
\newcommand*{\nats}{\ensuremath{\mathbb{N}}}

\newcommand*{\field}{\ensuremath{\mathbb{F}}}

\newcommand*{\bigo}[1]{\ensuremath{\mathcal{O}(#1)}}

\newcommand*{\enc}[1]{\ensuremath{[ #1 ]}}

\newcommand*{\partysymbol}[1]{\ensuremath{P_{#1}}}
\newcommand*{\parties}{\ensuremath{\mathcal{P}}}

\newcommand*{\numparties}{\ensuremath{N}}

\newcommand*{\prot}{\ensuremath{\pi}}
\newcommand*{\func}{\ensuremath{\pazocal{F}}}
\newcommand*{\gat}{\ensuremath{\rho}}

\AtBeginDocument{%
  }

\setcopyright{none}
\renewcommand\footnotetextcopyrightpermission[1]{}

\begin{document}

\title{Efficient Privacy-Preserving Approximation of the Kidney Exchange Problem}


\author{Malte Breuer}
\affiliation{%
  \institution{RWTH Aachen University}
  \city{Aachen}
  \country{Germany}}
\email{breuer@itsec.rwth-aachen.de}

\author{Ulrike Meyer}
\affiliation{%
	\institution{RWTH Aachen University}
	\city{Aachen}
	\country{Germany}}
\email{meyer@itsec.rwth-aachen.de}

\author{Susanne Wetzel}
\affiliation{%
	\institution{Stevens Institute of Technology}
	\city{Hoboken, NJ}
	\country{USA}}
\email{swetzel@stevens.edu}

\renewcommand{\shortauthors}{Breuer et al.}

\begin{abstract}
	The kidney exchange problem (KEP) seeks to find possible exchanges among pairs of patients and their incompatible kidney donors while meeting specific optimization criteria such as maximizing the overall number of possible transplants. 
Recently, several privacy-preserving protocols for solving the KEP have been proposed. However, the protocols known to date lack scalability in practice since the KEP is an NP-complete problem. We address this issue by proposing a novel privacy-preserving protocol which computes an approximate solution for the KEP that scales well for the large numbers of patient-donor pairs encountered in practice.
As opposed to prior work on privacy-preserving kidney exchange, our protocol is generic w.r.t.\ the security model that can be employed.
Compared to the most efficient privacy-preserving protocols for kidney exchange existing to date, our protocol is entirely data oblivious and it exhibits a far superior run time performance. 
As a second contribution, we use a real-world data set to simulate the application of our protocol as part of a kidney exchange platform, where patient-donor pairs register and de-register over time, and thereby determine its approximation quality in a real-world setting.
\end{abstract}



\maketitle
\pagestyle{empty}

\copyrightnotice{\copyright\space Copyright held by the owner/author(s) 2024. This is the author's version of the work. It is posted here for your personal use. Not for redistribution. The definitive version will be published in the \emph{ACM Asia Conference on Computer and Communications Security (ASIA CCS~’24), July 1–5, 2024, Singapore, Singapore.}}

\section{Introduction}
As of today, in the USA alone there are about 93,000 patients on the waiting list for a post-mortem kidney transplant with an average waiting time of four years~\cite{UNOS_WaitingList_2023}. Receiving a living kidney donation is an alternative to waiting for a post-mortem kidney transplant.
The main challenge with living kidney donation is to find organ donors who are medically compatible with the patients. 

A solution to this problem is kidney exchange where constellations between multiple patients with willing but incompatible donors (so-called \emph{patient-donor pairs}) are determined such that the patients can exchange their kidney donors among each other. 
The problem of finding a constellation that is optimal w.r.t.\ certain pre-defined criteria is known as the \emph{Kidney Exchange Problem}~(KEP). 

The KEP is typically modeled as a graph problem where a node~$ v $ is introduced for each patient-donor pair $ \partysymbol{v} $. An edge~$ (u, v) $ of weight $ \weightFunction(u, v) > 0 $ indicates that the donor of pair $ \partysymbol{u} $ is medically compatible with the patient of pair $ \partysymbol{v} $, where the weights correspond to the criteria under which the KEP is to be solved. These criteria may be used for the prioritization of certain exchanges (e.g., where patient and donor are of similar age). Solving the KEP then corresponds to finding a constellation of disjoint cycles in this \emph{compatibility graph} such that the sum of the weights of the edges in the constellation of disjoint cycles is maximized. Exchanges can only be carried out in form of \emph{exchange cycles} to guarantee that a donor of a pair donates her kidney only if the associated patient also receives a compatible kidney in return. In practice, the maximum allowed size of exchange cycles is typically limited to three since for larger cycles too many medical resources are needed simultaneously as all transplants in a cycle must be performed simultaneously to prevent a donor from dropping out at a time when the associated patient already received a kidney transplant~\cite{AshlagiKidneyExchangeOperations2021}.

The KEP is a well-researched problem in the field of operations research (e.g.,~\cite{AbrahamClearingAlgorithms2007,AndersonTSP2015,Dickerson_Chains_2012,Constantino_KEP_2013,RothEfficientKidneyExchange2007}) and in many countries there are already kidney exchange platforms that facilitate the computation of exchanges among the registered patient-donor pairs~\cite{AshlagiKidneyExchangeOperations2021,Biro_EuropeanModellingKE_2019}. However, there are also countries, such as Germany~\cite{TPG} or Brazil~\cite{Bastos_KEBrazil_2021,Roth_IllegalCountries_2022}, where kidney exchange is not allowed by law due to the fear of manipulation, corruption, and coercion. 

In countries, where kidney exchange is possible, the hospitals typically upload the medical data of their registered patient-donor pairs to the kidney exchange platform, where the KEP is solved repeatedly at certain pre-defined time intervals by the operator of the kidney exchange platform.
This centralized setting has two major security problems. (1) The platform operator is in complete control of the exchange computation and (2) the medical data of all patient-donor pairs is available in plaintext at a single entity.

This is particularly problematic as an attacker who can corrupt the platform operator may arbitrarily manipulate the computation of exchanges, e.g., enforcing that a particular patient is treated with priority.  
Also, in this centralized setting any attack leading~to~a~data breach has extremely severe consequences as it involves the highly sensitive medical data of many patients and donors. 
For the use case of organ donation this is a serious issue as trust in the platform is crucial to ensure the participation of as many donors as possible.
Overall, the platform operator is not only a prime target for high-impact attacks but also susceptible to manipulation and corruption.

\paragraph{Recent Advances}
To mitigate these security issues raised by existing centralized kidney exchange platforms, several privacy-preserving protocols for solving the KEP have been proposed~\cite{Breuer_KEprotocol_2020,Breuer_Matching_2022,Breuer_KepIp_2022,Birka_PPKidneyExchange_2022}. These protocols pursue a decentralized approach using \emph{Secure Multi-Party Computation} (SMPC). As such, they guarantee that all sensitive medical data is only stored in a decentralized fashion such that no single entity obtains access to the plaintext data of others. 
All of these protocols consider the semi-honest security model only, where the adversary is assumed to strictly follow the protocol specification. While this limits the impact of a data breach, it may not provide sufficient protection against manipulation of the computation of exchanges.
Besides, even the most efficient approaches known to date~\cite{Birka_PPKidneyExchange_2022,Breuer_KepIpExtended_2022} only scale for small numbers of patient-donor pairs and they are not data oblivious. More specifically, \cite{Birka_PPKidneyExchange_2022}, which computes an approximate solution to the KEP, reveals the number of cycles in the compatibility graph, and \cite{Breuer_KepIp_2022}, which computes an exact solution to the KEP, reveals the number of Simplex iterations required by their integer programming approach.

\paragraph{Contributions}
We propose a novel privacy-preserving protocol that is fully data oblivious and efficiently computes an approximate solution for the KEP.
In contrast to the existing protocols~\cite{Birka_PPKidneyExchange_2022,Breuer_KEprotocol_2020,Breuer_KepIp_2022,Breuer_Matching_2022}, we keep the security model for our protocol generic, using an \emph{arithmetic black box}. This allows us to utilize different SMPC primitives (also called \emph{general purpose SMPC protocols}). 
Depending on the level of security provided by the underlying SMPC primitive, our protocol exhibits security in the semi-honest or malicious model (where in the latter, the adversary may arbitrarily deviate from the protocol specification) considering either an honest or a dishonest majority. Thus, in the malicious model, our protocol also protects against any manipulation of the computation of exchanges. 

We have evaluated the run time of our protocol for different SMPC primitives using the framework MP-SPDZ~\cite{KellerMPSPDZ2020}. We show that our protocol achieves significantly superior run times both compared to the approximation protocol from~\cite{Birka_PPKidneyExchange_2022} and the fastest exact protocol from~\cite{Breuer_KepIp_2022}. For example, in \cite{Birka_PPKidneyExchange_2022} it is estimated that their protocol runs for more than 24 hours for 40 patient-donor pairs considering the semi-honest model with dishonest majority, whereas our protocol finishes within 2 minutes for the same security model. Similarly, the protocol in \cite{Breuer_KepIp_2022} runs for more than 105 minutes for 40~pairs in the semi-honest model with honest majority whereas our protocol finishes within 4 seconds for this security model and within 12 seconds for the malicious model with honest majority.

Furthermore, we have evaluated the approximation quality both of our novel protocol and the protocol from~\cite{Birka_PPKidneyExchange_2022} using a real-world data set from the United Network for Organ Sharing (UNOS), which is a large kidney exchange platform in the USA.\footnote{The data reported here have been supplied by the United Network for Organ Sharing as the contractor for the Organ Procurement and Transplantation Network. The~interpretation and reporting of these data are the responsibility of the author(s) and~in~no way should be seen as an official policy of or interpretation by the OPTN or the U.S. Government.} 
Our evaluation shows that despite its theoretical approximation ratio of $ 1/3 $, our protocol achieves an average approximation quality of about $ 80 $\% for a single run on the real-world data set from UNOS. Compared to~\cite{Birka_PPKidneyExchange_2022}, our protocol achieves a superior approximation quality for up to $ 30 $ patient-donor pairs. Note, however, that the run time of~\cite{Birka_PPKidneyExchange_2022} exceeds $ 24 $ hours for $ 40 $ pairs and more, and that run times longer than $ 24 $ hours are unacceptable in practice since there are kidney exchange platforms that execute match runs on a daily basis~\cite{AshlagiKidneyExchangeOperations2021}.

While an average approximation quality of $ 80 $\% for a single run may still seem too far from the optimal solution, it is important to recall that in practice the existing kidney exchange platforms solve the KEP repeatedly. Thus, not matching two patient-donor pairs in one run does not mean that these cannot be matched to other pairs in a future run. We investigate the impact of these dynamics in detail by means of simulations based on the real-world data set from UNOS. This allows us to determine the number of transplants found over time when using our approximation protocol in comparison to using an approach that computes the optimal solution in each match run. Our simulations show that the difference in the number of transplants found over time comparing our approach and the non-privacy-preserving approach is very small in practice. 

In countries where kidney exchange is already established, any loss in the number of transplants implied by our approach may not be acceptable. However, in countries such as Germany or Brazil, where security concerns currently prohibit kidney exchange, our approach may be the only viable option to pave the way for establishing kidney exchange in the future.
\section{Preliminaries}

In this section, we introduce relevant terminology in the context of Secure Multi-Party Computation (SMPC) (Section~\ref{sub:smpc}) and kidney exchange (Section~\ref{sub:kidney_exchange_terminology}) as well as the considered setup for privacy-preserving kidney exchange (Section~\ref{sub:pre_ppke_setup}).

\subsection{Secure Multi-Party Computation}\label{sub:smpc}

An SMPC protocol allows a set of parties to compute some functionality in a distributed fashion such that each party only learns its private input and output and what can be deduced from both. 

To ensure that our protocol does not depend on the use of a specific SMPC primitive, we consider an \emph{arithmetic black box}. This means that we allow the use of any SMPC primitive that is based on a secret-sharing scheme over the ring $ \ints_{2^k} $ (for a large integer $ k $) and that enables the local addition of secret values as well as the multiplication of secret values using a distributed protocol.\footnote{It is also possible to consider SMPC primitives over the field $ \field_p $. However, we only use the ring setting since it yields faster run times for our protocol.} In the following, we denote secret values by $ \enc{x} $, the $ i $-th entry of a vector $ \enc{V} $ of secret~values by $ \enc{V}(i) $, and the entry in the $ i $-th row and $ j $-th column of a matrix $ \enc{M} $ of secret values~by~$ \enc{M}(i, j) $. 

\paragraph{Security Model}
Depending on the used SMPC primitive, we obtain different levels of security for our privacy-preserving protocol for kidney exchange.  
We distinguish between the semi-honest and the malicious model. The former assumes that the corrupted parties strictly follow the protocol specification whereas in the latter the parties may arbitrarily deviate from the protocol specification.
Furthermore, we distinguish between an honest and a dishonest majority. The former assumes that less than half of the parties are corrupted and in the latter any number of parties may be corrupted as long as one party is honest. 
Note that our protocol for kidney exchange consists of a sequence of calls to gates for which there exist secure implementations in all of these security models. Following the composition theorem from~\cite{Canetti_UC_2001}, our protocol thus achieves the same level of security as provided by the used SMPC primitive.

\paragraph{Complexity Metrics}
We consider the two complexity metrics communication and round complexity, where the communication complexity is determined by the overall amount of data sent during a
protocol execution and the round complexity describes the number of sequential communication steps
required by a protocol. 

\paragraph{Gates}

In our privacy-preserving protocols, we use several existing gates for the computation of basic operations. One such operation is the comparison of two secret values. Damgård et al.~\cite{Damgard_NewPrimitivesActiveSecurity_2019} show how such a gate can be constructed with communication complexity $ \bigo{\bitsize} $ and round complexity $ \bigo{\log \bitsize} $, where $ \bitsize $ is the bitsize of secret-shared values. 
We also sometimes need to choose between two values $ \enc{x} $ and $ \enc{y} $ based on a secret bit $ \enc{z} $. We denote this gate by $ \enc{z} \ ? \ \enc{x} \ : \ \enc{y} $ where the output is $ \enc{x} $ if $ z = 1 $ and $ \enc{y} $ otherwise. It can be implemented by computing $ \enc{z} \cdot (\enc{x} - \enc{y}) + \enc{y} $ and thus requires a single multiplication. 
We introduce a gate $ \sharedDotProduct $ that computes the dot product of two secret vectors of size $ n $ requiring $ n $ multiplications and one round of communication.
The gate $ \demuxGate{n} $ is used for de-multiplexing a value $ x $ into a vector~of~size $ n $ that contains $ \enc{1} $ in the $ x $-th entry and $ \enc{0} $ for all other entries. This can be implemented using the approach from~\cite{Launchbury_Multiplex_2012} with communication complexity $ \bigo{2^\bitsize} $ and round complexity $ \bigo{\log \bitsize} $ (cf.~Appendix~\ref{app:gates}).
We use a gate $ \shuffleGate $ for obliviously shuffling the entries of a secret vector of size $ n $. We use the implementation from~\cite{Keller_EfficientDataStructures_2014} yielding communication complexity $ \bigo{n \log n} $ and round complexity $ \bigo{\log n} $.\footnote{While there are more efficient shuffling implementations (e.g.,~\cite{Araki_SecureGraphAnalysis_2021}), we use the one from~\cite{Keller_EfficientDataStructures_2014} as it is contained in MP-SPDZ and shuffling has a small impact on the run time of our protocol (e.g., $ 6\% $ of the overall runtime for 5 pairs, and $ 0.02\% $ for 200~pairs).} A gate $ \reverseShuffleGate $ that reverts the shuffling with permutation $ \perm $ can be implemented analogously~using~$ \perm^{-1} $. 

\subsection{Kidney Exchange Terminology}\label{sub:kidney_exchange_terminology}

The goal of kidney exchange is to find constellations in which a set of patient-donor pairs $ \partysymbol{i} $ ($ i \in \parties $) with $ \parties = \{0, ..., \numparties - 1\} $ can exchange their kidney donors among each other such that certain pre-defined optimization criteria are met. The medical compatibility between the pairs is modeled by a so-called \emph{compatibility graph}.

\begin{definition}[Compatibility graph]
	The compatibility graph $ \compGraph $ is defined as the directed graph $ \compGraph = (V, E, \weightFunction) $ containing a node for each patient-donor pair $ \partysymbol{i} $ (i.e., $ V = \{0, ..., \numparties - 1\} $). An edge $ (i, j) \in E $ indicates that the donor of pair $ \partysymbol{i} $ is medically compatible with the patient of pair $ \partysymbol{j} $ ($ i \neq j $) and the weight function $ \weightFunction : E \rightarrow \nats $ maps each edge $ (i, j) \in E $ to an integer weight that is used for prioritization.  
\end{definition}

We represent such a graph by an $ \numparties \times \numparties $ adjacency matrix $ M $ and an $ \numparties \times \numparties $ prioritization matrix $ \prioMatrix $ where $ M(i, j) = 1 $ iff $ (i, j) \in E $ and $ M(i, j) = 0 $, otherwise, and $ \prioMatrix(i, j) = \weightFunction(i, j) $.

In kidney exchange it must be ensured that the donor of a pair only donates her kidney if the associated patient also receives a compatible kidney transplant in return. Thus, exchanges among patient-donor pairs can only take place in form of \emph{exchange cycles}. 

\begin{definition}[Exchange cycle]
	For a set of patient-donor pairs $ \{\partysymbol{0}, ..., \partysymbol{\numparties-1}\} $, an exchange cycle $ \singleCycle $ of size $ \cycleBound $ (also called a $ \cycleBound $-cycle) is a tuple $ (v_1, .., v_{\cycleBound}) $ with $ v_i \neq v_j $ for $ i \neq j $ such that the donor of $ \partysymbol{v_i} $ is medically compatible with the patient of $ \partysymbol{v_{i+1}} $ for $ i \in \{1, ..., \cycleBound-1\} $ and the donor of $ \partysymbol{v_\cycleBound} $ is compatible with the patient of $ \partysymbol{v_1} $.
\end{definition}

The \emph{Kidney Exchange Problem} (KEP) then consists of finding a set of exchange cycles that is optimal w.r.t.\ certain pre-defined criteria. 

\begin{definition}[Kidney Exchange Problem (KEP)]\label{def:kep}
	Given a compatibility graph $ \compGraph $ for a set of patient-donor pairs $ \partysymbol{i} $ ($ i \in \{0, ..., \numparties-1\} $), the kidney exchange problem is defined as finding a set of vertex disjoint exchange cycles $ \{\singleCycle_{1}, ..., \singleCycle_{l}\} $ in $ \compGraph $ such that the sum of the edge weights of these cycles is maximized.
\end{definition}

Note that all transplants of a cycle have to be executed simultaneously to avoid the situation where a donor retreats from the exchange after the associated patient already received a kidney transplant. Therefore, the KEP is typically considered for a maximum cycle size $ \cycleBound = 3 $ as larger cycles require the simultaneous allocation of too many medical resources~\cite{AshlagiKidneyExchangeOperations2021,Biro_EuropeanModellingKE_2019}. Thus, we also consider $ \cycleBound = 3 $ in our protocol for computing an approximate solution to the KEP~(cf.~Section~\ref{sub:protocol_spec}). Note that the KEP is NP-complete if the maximum cycle size is restricted to any value larger than two~\cite{AbrahamClearingAlgorithms2007}. 

\begin{figure}
	\includegraphics*[scale=1.0]{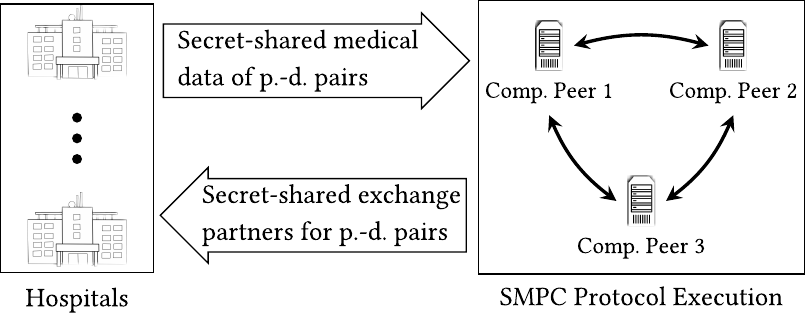}
	\caption{Common setup of a privacy-preserving kidney exchange platform where the hospitals secretly share the data of their patient-donor pairs (p.-d.\ pairs) among the computing peers who use SMPC to determine the exchange partners.}
	\label{fig:pp_platform}
	\vspace*{0.5em}
\end{figure}

\subsection{Privacy-Preserving Kidney Exchange Setup}\label{sub:pre_ppke_setup}

The common setup for privacy-preserving kidney exchange~\cite{Breuer_Matching_2022,Breuer_KepIp_2022,Birka_PPKidneyExchange_2022} is the \emph{client-server} model, in which the hospitals (clients), where the patient-donor pairs are registered, use secret-sharing to provide the private medical data of their patient-donor pairs to a set of computing peers (servers) that execute the actual SMPC protocol which determines potential exchange partners for the patient-donor pairs. The hospitals otherwise do not actively participate in the protocol execution. While the computing peers exchange messages among each other during the protocol execution, the use of SMPC guarantees that these messages do not leak any information. Thus, the computing peers do not learn anything from the protocol execution. Possible candidates for hosting the computing peers are governmental institutions, universities, or large transplant centers. After executing the SMPC protocol, the computing peers use secret-sharing to inform the hospitals about the computed potential exchange partners for their patient-donor pairs. The final decision of whether the computed transplants are carried out then lies with medical experts at the hospitals. Based on~\cite{Breuer_Matching_2022,Breuer_KepIp_2022,Birka_PPKidneyExchange_2022}, Figure~\ref{fig:pp_platform} shows this setup for the example of three computing peers.

\paragraph{Security Model}
In contrast to the related work in the field of privacy-preserving kidney exchange~\cite{Breuer_KEprotocol_2020,Breuer_Matching_2022,Breuer_KepIp_2022,Birka_PPKidneyExchange_2022}, we do not limit our protocol to the semi-honest model but also consider the malicious security model. Similarly, our protocol can be instantiated with a honest majority as well as with a dishonest majority of computing peers. Consequently, when using our protocol as part of a privacy-preserving kidney exchange platform, one can then decide on the particular security model to be considered. In general, it holds that the run time performance of the protocol decreases when considering the malicious model or a dishonest majority. Thus, the decision for a particular security model forms a trade-off between run time performance and the provided security model.

The assumption in the semi-honest model is that the corrupted parties strictly follow the protocol specification. Therefore, this model does not fully protect against an attacker who tries to manipulate the computation of the exchanges. However, the semi-honest model already ensures that a data breach at a single computing peer does not have any consequences as the shares that are present at a single computing peer do not reveal any information on the underlying sensitive medical data. If protection against any manipulation of the computation of the exchanges is desired, the malicious model has to be employed. This guarantees that an attacker is not able to manipulate the computation as long as he controls less than half of the computing peers (honest majority) or if at least one computing peer is not controlled by the attacker (dishonest majority). 
\section{Related Work}

In this section, we review related work in conventional and privacy-preserving kidney exchange (cf.~Section~\ref{sub:rw_convKE} and Section~\ref{sub:rw_ppKE}), as well as privacy-preserving approximation protocols (cf.~Section~\ref{sub:rw_ppApprox}).

\subsection{Conventional Kidney Exchange}\label{sub:rw_convKE}
The most efficient conventional (non-privacy-preserving) algorithms that compute an exact solution for the KEP use integer programming (IP)~\cite{AshlagiKidneyExchangeOperations2021,Biro_EuropeanModellingKE_2019}. This technique has also already been explored in the privacy-preserving setting but was shown to scale only for small numbers of patient-donor pairs~\cite{Breuer_KepIp_2022}.
Although the IP-based approaches efficiently compute an optimal solution in the conventional setting, some kidney exchange platforms still use an approximative approach~\cite{Biro_EuropeanModellingKE_2019}. 
In particular, Spain~\cite{Bofill_SpanishKEP_2017} uses a Greedy approach that chooses the $ 2 $-cycle of maximum weight until there are no more $ 2 $-cycles and only then starts choosing the $ 3 $-cycle of maximum weight. Thereby, cycles of size $ 2 $ are strictly prioritized over cycles of size $ 3 $. In contrast, the Greedy algorithm (cf.~Algorithm~\ref{alg:greedy}) that we follow in our protocol repeatedly chooses the cycle of maximum weight among all cycles of size $ 2 $ and $ 3 $. Note, however, that the Spanish approach can also be modeled in our approach by choosing the weight function such that it strictly prioritizes cycles of size $ 2 $ over cycles of size $ 3 $.
After executing the Greedy algorithm, the Spanish approach further optimizes the solution by increasing the number of $ 3 $-cycles that include a $ 2 $-cycle and by checking for each $ 3 $-cycle whether it can be replaced by two cycles of size $ 2 $.

\subsection{Privacy-Preserving Kidney Exchange}\label{sub:rw_ppKE}

To the best of our knowledge, there are two lines of work in privacy-preserving kidney exchange. The first comprises the SMPC protocols by Breuer et al.~\cite{Breuer_KEprotocol_2020,Breuer_Matching_2022,Breuer_KepIp_2022} that compute exact solutions for the KEP, considering the semi-honest model with honest~majority. 

The protocol presented in~\cite{Breuer_Matching_2022} solves the KEP for cycles of size $ 2 $ only. It relies on a privacy-preserving implementation of the matching algorithm by Pape and Conradt~\cite{PapeMatching1980} and it has communication complexity $ \bigo{\numparties^5} $ and round complexity $ \bigo{\numparties^4} $. This results in considerably high run times. E.g., the protocol runs for more than $ 5 $ days for $ 65 $ patient-donor pairs. However, the authors show that when used in a dynamic setting where pairs register and de-register over time, in practice the protocol nearly achieves the same number of transplants as the algorithmic approach used in centralized platforms. As the protocol considers the special case of $ 2 $-cycles, we do not compare our results to this protocol. However, we use the simulation framework from~\cite{Breuer_Matching_2022} for simulating our protocol~$ \kepApProtocol $ when used in a dynamic kidney exchange platform (cf.~Section~\ref{sec:simulation}). 

Breuer et al.\ also introduced two protocols that compute exact solutions for the KEP for cycles of size $ 2 $ and $ 3 $. The protocol from~\cite{Breuer_KEprotocol_2020} uses homomorphic encryption and follows a brute force approach that compares the secret compatibility graph to all possible existing constellations of exchange cycles. As the set of exchange constellations increases quickly for increasing numbers of patient-donor pairs, their protocol only exhibits feasible run times for a small number of pairs (e.g., it takes more than $ 13 $ hours for $ 9 $ pairs). 

In a subsequent work, Breuer et al.~\cite{Breuer_KepIp_2022} propose the protocol~$ \kepIpProtocol $, which is a more efficient SMPC protocol based on secret sharing for solving the KEP using integer programming. The protocol~$ \kepIpProtocol $ also computes an exact solution for the KEP but it is not data oblivious as its flow of execution depends on the number of iterations of the Simplex algorithm, which is used as a subroutine. The run times of the protocol~$ \kepIpProtocol $ are significantly better than those of the protocol from~\cite{Breuer_KEprotocol_2020}. Since this is the fastest SMPC protocol that exactly solves the KEP for cycles of size $ 2 $ and $ 3 $ to date, we compare the run time of our new protocol~$ \kepApProtocol $ against those of the protocol $ \kepIpProtocol $ to determine the performance gain of our approximation approach (cf.~Section~\ref{sub:eval_runtime}). We also simulate the use of $ \kepIpProtocol $ in a dynamic setting and compare the results to those obtained for our protocol~$ \kepApProtocol $ when used in a dynamic setting (cf.~Section~\ref{sub:sim_results}). The simulations show that using our protocol outperforms the use of~$ \kepIpProtocol $ for nearly all parameter constellations.

The second line of work is the approximation protocol by Birka et al.~\cite{Birka_PPKidneyExchange_2022}. They consider two computing peers and implement their protocol in the ABY framework~\cite{Demmler_ABY_2015}. Thus, their protocol is secure in the semi-honest model with dishonest majority.
Although~\cite{Birka_PPKidneyExchange_2022} also computes an approximation of the KEP, their underlying algorithmic approach is entirely different from Algorithm~\ref{alg:greedy}, which forms the basis for our protocol~$ \kepApProtocol $. 
Comparing the computed set of exchange cycles, there are two major differences between the two approximation approaches.
First,~\cite{Birka_PPKidneyExchange_2022} can only approximate the KEP for a single cycle size while our protocol considers all cycles up to the maximum cycle size (as is common when solving the KEP in practice~\cite{Biro_EuropeanModellingKE_2019,AshlagiKidneyExchangeOperations2021}).
Second, we directly apply a Greedy approximation on the set of~all potential exchange cycles, repeatedly adding the cycle of maximum weight to the solution. 
In contrast,~\cite{Birka_PPKidneyExchange_2022} first reduces the set of all potential cycles to the ones that exist in the compatibility graph.
Then, they compute one separate solution for each cycle in the graph, by first adding that cycle to the solution and then adding further cycles using a Greedy approach. 
Finally, they choose the solution that yields the highest weight among all computed solutions. 
Compared to our protocol, the protocol from~\cite{Birka_PPKidneyExchange_2022} yields significantly worse run times. For example, in~\cite{Birka_PPKidneyExchange_2022} it is estimated that their protocol runs for more than $ 24 $ hours for $ 40 $ pairs and cycle size $ 3 $ whereas our protocol terminates within $ 2 $~minutes, considering cycles of size $ 2 $ and $ 3 $ and the same security model as~\cite{Birka_PPKidneyExchange_2022}.\footnote{Due to its high RAM consumption \cite{Birka_PPKidneyExchange_2022} extrapolates the run times of their protocol once the number of patient-donor pairs is larger than $ 18 $.}
Thus, in contrast to our protocol, the run times of~\cite{Birka_PPKidneyExchange_2022} do not scale for the numbers of patient-donor pairs encountered in practice (e.g., around 200 pairs for UNOS). 
We attribute this to the different algorithmic approaches of the two protocols.
It is important to note that the protocol from~\cite{Birka_PPKidneyExchange_2022} is not data oblivious as its flow of execution depends on  the number of cycles in the private compatibility graph. Thus, the protocol run time also depends on the number of cycles in the graph. 
Finally, our evaluation of the approach from~\cite{Birka_PPKidneyExchange_2022} (Section~\ref{sub:eval_quality}) shows that the approximation quality of~\cite{Birka_PPKidneyExchange_2022} is worse than for our approach for up to $ 30 $ patient-donor pairs.\footnote{Note that \cite{Birka_PPKidneyExchange_2022} only evaluates the run time performance of their protocol but not the quality of the approximation.} While for more pairs the quality of~\cite{Birka_PPKidneyExchange_2022} is on average about $ 3.9 $\% better than for our approach, recall that \cite{Birka_PPKidneyExchange_2022} takes more than $ 24 $ hours for $ 40+ $ pairs. This is unacceptable in practice as there are kidney exchange platforms that execute match runs on a daily~basis~\cite{AshlagiKidneyExchangeOperations2021}.

\subsection{Privacy-Preserving Approximation Protocols for Graph Problems}\label{sub:rw_ppApprox}

Brüggemann et al.~\cite{Bruggemann_MaxMatchingApproximation_2022} present an SMPC protocol for approximating a maximum weight matching on a general graph. Their protocol also follows a Greedy approach in that it repeatedly adds the edge of maximum weight to the solution and then removes all edges from the graph that are incident to the chosen edge. Thereby, their protocol achieves a $ 1/2- $approximation for the maximum weight matching problem on general graphs. For the special case that only cycles of size $ 2 $ are allowed, the KEP can be reduced to the maximum weight matching problem on general graphs. 
Thus, their protocol could be used as a building block for computing an approximate solution for the KEP for a maximum cycle size of $ 2 $. However, it cannot be used for computing an approximate solution for the KEP for larger cycles. In contrast to this, our protocol finds an approximate solution for the KEP for cycles of size up to $ 3 $, which is the common setting for kidney exchange in practice~\cite{AshlagiKidneyExchangeOperations2021,Biro_EuropeanModellingKE_2019}. 
\section{Privacy-Preserving Approximation of the Kidney Exchange Problem}\label{sec:protocol}

In this section, we present our novel protocol $ \kepApProtocol $ that computes an approximate solution for the KEP.

\subsection{Approach and Ideal Functionality}\label{sub:ideal_func}

Intuitively, one would design a Greedy algorithm for computing an approximate solution for the KEP by repeatedly adding the cycle of maximum weight to the solution until the compatibility graph no longer contains a cycle that is disjoint from all previously chosen cycles.
However, in the privacy-preserving setting we have to hide which cycles exist in the compatibility graph. Thus, we have to consider all \emph{potential cycles} that could exist between the patient-donor pairs, i.e., all cycles in a complete graph with $ \numparties $ nodes. 
This yields $ \sum_{i = 2}^{3} \frac{\numparties !}{(\numparties - i)! \cdot i} $ potential cycles for a maximum cycle size of~$ 3 $.

We can reduce the size of this set to $ \sum_{i = 2}^{3} \frac{\numparties !}{(\numparties - i) ! \cdot i !} $ by only considering all different subsets of nodes that could form an exchange cycle. In particular, each subset $ \{u, v, w\} $ of three nodes can yield two different cycles $ (u, v, w) $ and $ (u, w, v) $. 
Thus, instead of considering the set of all potential cycles in each iteration, we can first evaluate the cycle of maximum weight for each subset and then iteratively add the subset of maximum weight to the solution. 

Algorithm~\ref{alg:greedy} shows the Greedy approach for computing an approximate solution for the KEP using the set of all subsets instead of exchange cycles. The algorithm takes the compatibility graph $ \compGraph $ and the set $ \subsets $ of all subsets of size $ 2 $ and $ 3 $ as input. The set $ \subsets $ is ordered such that it contains all subsets of size $ 3 $ followed by all subsets of size $ 2 $, each sorted in ascending order by their~nodes.

As a first step, the nodes of the compatibility graph are shuffled uniformly at random. This ensures that the index of a patient-donor pair does not influence the probability that the pair is included in the computed solution. 
Then, the maximum weight cycle for each subset is determined and stored in the map $ \subsetCycleMap $. 
In each iteration of the main loop, the first subset of maximum weight is determined and added to the solution set $ \chosenSets $. 
Afterwards, all subsets that share at least one node with the chosen subset $ \singleSubset_i $ are removed from~$ \subsets $. The main loop terminates when there are no more subsets of weight larger than $ 0 $ in $ \subsets $. As a last step, the chosen subsets $ \chosenSets $ are mapped back to the corresponding cycles~$ \chosenCycles $ using the previously computed map $ \subsetCycleMap $.

\begin{algorithm}[t]
	\algrenewcommand\algorithmicindent{0.7em}
\small
\begin{flushleft}
	\textbf{Input:} Compatibility graph $ \compGraph = (V, E, \weightFunction) $, subsets~$ \subsets $ of size $ 2 $ and~$ 3 $ \vspace*{0.2em} \\
	\textbf{Output:} Set $ \chosenCycles $ of disjoint exchange cycles from $ \cg $
\end{flushleft}
\begin{algorithmic}[1]
	\State Shuffle $ \cg $ uniformly at random
	\For{$ s \in \subsets $}
	\vspace*{-0.2em}
		\State Let $ \singleCycle_{\textit{max}} $ be the first cycle of max.\ weight for the nodes in $ s $
		\State $ \subsetCycleMap(s) \leftarrow \singleCycle_{\textit{max}} $
		\State $ \cycleWeight(s) \leftarrow \textit{weight}(\singleCycle_{\textit{max}}) $
	\EndFor
	\vspace*{-0.1em}
	\State $ \chosenSets \leftarrow \emptyset $
	\vspace*{-0.1em}
	\While{$ \exists s \in \subsets : \cycleWeight(s) > 0 $}
		\State Let $ \singleSubset_{\textit{max}} $ be the first subset of maximum weight in $ \subsets $
		\State $ \chosenSets \leftarrow \chosenSets \cup \{\singleSubset_{\textit{max}}\} $
		\State $ \subsets \leftarrow \subsets \setminus \{\singleSubset \in \subsets \ | \ \singleSubset \cap \singleSubset_{\textit{max}} \neq \emptyset \} $
	\vspace*{-0.1em}
	\EndWhile
	\vspace*{-0.1em}
	\State $ \chosenCycles \leftarrow \{\subsetCycleMap(s) \ | \ s \in \chosenSets\} $
	\vspace*{-0.2em}
\end{algorithmic}

	\caption{Greedy Approximation Algorithm for the KEP}
	\label{alg:greedy}
\end{algorithm}

Note that by always choosing the first subset of maximum weight, we inherently choose the larger cycle if there are multiple cycles of maximum weight. This is in line with the common criterion to always maximize the number of transplants in kidney exchange~\cite{AshlagiKidneyExchangeOperations2021,Biro_EuropeanModellingKE_2019}. If this behavior is not desired for some reason, one can simply shuffle the set $ \subsets $ uniformly at random before the while loop in Algorithm~\ref{alg:greedy}. Thereby, one chooses a subset uniformly at random from all subsets of maximum weight. In the privacy-preserving~case, this induces a negligible run time overhead as $ \subsets $ is publicly known and thus can be shuffled without any computation on secret values.

The theoretical approximation ratio of Algorithm~\ref{alg:greedy} is $ 1/3 $. Intuitively, this holds as choosing a cycle of maximum weight $ \weight_m $ in the approximate solution can at most prevent the choice of three alternative cycles of weight $ \weight \leq \weight_m $ in the optimal solution. Thus, the weight of the optimal solution is at most three times the weight of the approximate solution. A formal proof is given in Appendix~\ref{app:proof}.

The corresponding ideal functionality that is implemented by our novel protocol $ \kepApProtocol $ is defined as follows. 

\begin{functionality}[$ \kepApFunc $ - Approximation of the KEP]\label{func:kepAp}
	\textit{
		Let all computing peers hold their respective shares of the secret quotes $ \enc{\inputQuote{i}} $ containing the medical data of the patient-donor pairs $ \partysymbol{i} $ ($ i \in \parties $) that is relevant for computing their medical compatibility and the weights in the compatibility graph. Further, let $ \subsets $ be the publicly known set of subsets for the patient-donor pairs $ \partysymbol{i} $ ($ i \in \parties $), containing all subsets of size $ 3 $ followed by all subsets of size $ 2 $ sorted in ascending order by their nodes. Then, functionality~$ \kepApFunc $ is given~as 
		\begin{equation*}
			(\enc{d_0}, \enc{r_0}), ..., (\enc{d_{\numparties - 1}}, \enc{r_{\numparties - 1}}) \leftarrow \kepApFunc(\enc{\inputQuote{0}}, ..., \enc{\inputQuote{\numparties -1}}, \subsets),
		\end{equation*}  
		where $ (d_i, r_i) $ are the indices of donor and recipient for pair~$ \partysymbol{i} $ ($ i \in \parties $) for a set of exchange cycles of maximum cycle size $ \cycleBound = 3 $ computed as specified in Algorithm~\ref{alg:greedy}. 
	}
\end{functionality}

As the two common criteria for compatibility in kidney exchange are blood type and Human Leukocyte Antigen (HLA) compatibility~\cite{AshlagiKidneyExchangeOperations2021}, the quote~$ \inputQuote{i} $ has to include four binary indicator vectors: $ \bloodtypeIndicatorVector{i}{d} $ for the donor blood type, $ \bloodtypeIndicatorVector{i}{p} $ for the patient blood type, $ \antigenvec{i} $ for the donor's HLA, and $ \antibodyvec{i} $ for the patient's HLA antibodies.
The data required for computing the weights in the compatibility graph depends on the particular criteria that are used for prioritization. As these vary widely among existing kidney exchange platforms~\cite{AshlagiKidneyExchangeOperations2021,Biro_EuropeanModellingKE_2019}, we keep them generic for our protocol description. The criteria used in our run time evaluation are further discussed in Section~\ref{sub:eval_runtime}.

\subsection{Maximum Weight Subset Computation}\label{sub:heaviest_cycle}

Since our protocol~$ \kepApProtocol $ follows the Greedy approach outlined in Algorithm~\ref{alg:greedy}, it requires a gate $ \searchHeaviestCycle $ that finds the first subset $ \singleSubset_i $ in the set $ \subsets $ that has maximum weight.

Gate~\ref{gat:heaviest_cycle} provides the specification of gate~$ \searchHeaviestCycle $, which takes a set of $ n $ subsets as input. These are encoded by a secret vector $ \enc{\cycleIndices} $ storing the unique index of each subset and a secret matrix $ \enc{\kepApNodes} $, where $ \enc{\kepApNodes}(i) = \{\enc{u}, \enc{v}, \enc{w}\} $ encodes the nodes of the $ i $-th subset. The weights for each subset are stored in the secret vector $ \enc{\cycleWeight} $. The gate then returns the index and the nodes of a subset with maximum weight. If there are multiple such subsets, the one with the minimal index is returned. If there is no subset of weight larger than $ 0 $, the dummy subset $ (\enc{\numparties}, \enc{\numparties}, \enc{\numparties}) $ with the dummy index $ \enc{\vert \subsets \vert} $ is returned. Note that we use zero indexing and thus can use the values $ \numparties $ and $ \vert \subsets \vert $ as dummy values. 

To minimize the number of communication rounds, we use a tree reduction approach as outlined in~\cite{Catrina_PrimitivesSMPC_2010}. In each recursion, we first check if the size $ n $ of the set of subsets is equal to~$ 1 $. If this is the case, we have found a subset of maximum weight. We return this subset if it is valid, i.e., if its weight is larger than $ 0 $. Otherwise, we return the dummy subset $ (\enc{\numparties}, \enc{\numparties}, \enc{\numparties}) $ with index $ \enc{\vert \subsets \vert} $. Note that we keep the actual size of a subset (i.e., whether it comprises two or three nodes) hidden by always encoding a subset by three secret nodes. If the $ i $-th subset only contains two nodes, the third node (stored in $ \enc{\kepApNodes}(i, 2) $) contains the dummy value~$ \enc{\numparties} $. If $ n $ is still larger than $ 1 $, we have not yet found a subset of maximum weight. In that case, we compare each subset at index~$ 2i $ with its neighbor at index $ 2i + 1 $. If the weight of the subset with index~$ 2i $ is larger or equal to the weight of the subset with index $ 2i + 1 $, we select the first subset (i.e., the subset with index $ 2i $). Otherwise, we select the subset with index $ 2i + 1 $. We store the index of the selected subset in the vector $ \enc{\cycleIndices'} $, its nodes in the matrix $ \enc{\kepApNodes'} $, and its weight in the vector $ \enc{\cycleWeight'} $. Thereby, in each recursion we effectively half the set of subsets that still have to be evaluated. If this set is of odd size, we just store the values for the last subset of the current iteration as the last entry of $ \enc{\cycleIndices'} $, $ \enc{\kepApNodes'} $, and $ \enc{\cycleWeight'} $. Finally, we recursively call the gate $ \searchHeaviestCycle $ on the new set of subsets encoded by $ \enc{\cycleIndices'} $, $ \enc{\kepApNodes'} $, and $ \enc{\cycleWeight'} $. 

\begin{gate}[t]
	\algrenewcommand\algorithmicindent{0.7em}
\small
\begin{flushleft}
	\textbf{Input:} Set of $ n $ subsets with secret indices $ \enc{\cycleIndices} $, secret node sets $ \enc{\kepApNodes} $, and secret weights $ \enc{\cycleWeight} $ \vspace*{0.2em} \\
	\textbf{Output:} Secret index $ \enc{\cycleIndices}(i) $ and set of nodes $ \enc{\kepApNodes}(i) $ of a subset where $ i $ is the lowest index for which $ \cycleWeight > 0 $ is maximal, or $ \enc{\vert \subsets \vert} $ and $ (\enc{\numparties}, \enc{\numparties}, \enc{\numparties}) $ if $ \cycleWeight = 0 $ for all $ 0 \leq i < n $
\end{flushleft}
\begin{algorithmic}[1]
	\If{$ n = 1 $}
	\vspace*{-0.3em}
		\State $ \enc{\textit{valid}} \leftarrow \enc{\cycleWeight}(0) \sharedGreaterThan \enc{0} $
		\State $ \enc{\indexHeaviest} \leftarrow \enc{\textit{{valid}}} \ ? \ \enc{\cycleIndices}(0) \ : \ \enc{\vert \subsets \vert} $
		\State $ \enc{\heaviest} \leftarrow \enc{\textit{valid}} \ ? \ \enc{\kepApNodes}(0) \ : \ (\enc{\numparties}, \enc{\numparties}, \enc{\numparties}) $
		\State \Return $ \enc{\indexHeaviest}, \enc{\heaviest} $
		\vspace*{-0.1em}
	\Else
	\vspace*{-0.1em}
		\For{$ 0 \leq i < \lfloor n / 2 \rfloor $ \textbf{in parallel}}
			\State $ \enc{\textit{first}} \leftarrow (\enc{\cycleWeight}(2 i) + \enc{1}) \sharedGreaterThan \enc{\cycleWeight}(2 i + 1) $
			\State $ \enc{\kepApNodes'}(i) \leftarrow \enc{\textit{first}} \ ? \ \enc{\kepApNodes}(2i) \ : \ \enc{\kepApNodes}(2i + 1) $
			\State $ \enc{\cycleIndices'}(i) \leftarrow \enc{\textit{first}} \ ? \ \enc{\cycleIndices}(2i) \ : \ \enc{\cycleIndices}(2i + 1) $
			\State $ \enc{\cycleWeight'}(i) \leftarrow \enc{\textit{first}} \ ? \ \enc{\cycleWeight}(2i) \ : \ \enc{\cycleWeight}(2i + 1) $
		\EndFor
		\If{$  \lceil n / 2 \rceil > \lfloor n / 2 \rfloor $}	
			\State $ \enc{\kepApNodes'}( \lceil n / 2 \rceil - 1) \leftarrow \enc{\kepApNodes}(n - 1) $
			\State $ \enc{\cycleIndices'}( \lceil n / 2 \rceil - 1) \leftarrow \enc{\cycleIndices}(n - 1) $
			\State $ \enc{\cycleWeight'}( \lceil n / 2 \rceil - 1) \leftarrow \enc{\cycleWeight}(n - 1) $
		\EndIf
		\vspace*{-0.2em}
		\State \Return $ \searchHeaviestCycle(\enc{\cycleIndices'}, \enc{\kepApNodes'}, \enc{\cycleWeight'}) $
		\vspace*{-0.1em}
	\EndIf
\end{algorithmic}
\vspace*{-0.1em}

	\caption{$ \searchHeaviestCycle $}
	\label{gat:heaviest_cycle}
\end{gate}

\paragraph{Security}
The gate $ \searchHeaviestCycle $ is fully data oblivious, i.e., its flow of execution is independent of the input and it only uses gates that have already been proven secure.
The gate can thus be simulated by applying the composition theorem from~\cite{Canetti_UC_2001} and following the publicly known flow of execution.

\paragraph{Complexity}
The gate has communication complexity $ \bigo{n \bitsize} $ and round complexity $ \bigo{\log n \log \bitsize} $, where $ n $ is the number of subsets.

\begin{protocolDouble}[t]
	\algrenewcommand\algorithmicindent{0.7em}
\small
\begin{flushleft}
	\textbf{Input:} Each hospital shares the input quote $ \enc{\inputQuote{i}} $ for each of its registered patient-donor pairs $ \partysymbol{i} $ ($ i \in \parties $) with the computing peers. Each computing peer also holds a public set $ \subsets $ of all potential subsets of $ \parties $ of size $ 2 $ and $ 3 $  and the corresponding node set $ \kepApNodes $ with $ \kepApNodes(i) = (u, v, w) $ for $ i \in \{0, ..., \vert \subsets \vert\} $ \vspace*{0.2em} \\
	\textbf{Output:} Each computing peer sends the shares of the exchange partners $ \enc{D}(i) $ and $ \enc{R}(i) $ to the corresponding hospital of patient-donor pair $ \partysymbol{i} $ ($ i \in \parties $)
\end{flushleft}
\vspace*{-1.4em}
\begin{multicols}{2}
	\begin{algorithmic}[1]
		\Statex \underline{Construction Phase}
		\For{$ 0 \leq i \neq j < \numparties $ \textbf{in parallel}}
			\State $ \enc{\compGraphMatrix}(i, j) \leftarrow \compCheck(\enc{\inputQuote{i}}, \enc{\inputQuote{j}}) $
			\vspace*{0.15em}
			\State $ \enc{\prioMatrix}(i, j) \leftarrow \prioGate(\enc{\inputQuote{i}}, \enc{\inputQuote{j}}) $
		\EndFor
		\vspace*{0.15em}
		\State $ \enc{\compGraphMatrix'} \leftarrow \shuffleGate(\enc{\compGraphMatrix}, \enc{\shufflePermutation}) $ 
		\vspace*{0.1em}
		\State $ \enc{\prioMatrix'} \leftarrow \shuffleGate(\enc{\prioMatrix}, \enc{\shufflePermutation}) $
		\vspace*{0.1em}
		\For{$ 0 \leq i < \vert \subsets \vert $}
			\State $ \enc{\chosenSets}(i) \leftarrow \enc{0} $
			\State $ \enc{\cycleIndices}(i) \leftarrow \enc{i} $
		\EndFor
		\Statex \underline{Evaluation Phase}
		\For{$ 0 \leq i < \vert \subsets \vert $ \textbf{in parallel}}
			\State $ u, v, w \leftarrow \kepApNodes(i) $
			\If{$ w = \numparties $}
				\State $ \enc{\cycleWeight}(i) \leftarrow \enc{\compGraphMatrix}(u, v) \sharedMult \enc{\compGraphMatrix}(v, u)  $
				\Statex \hspace*{6.9em} $ \sharedMult (\enc{\prioMatrix}(u, v) + \enc{\prioMatrix}(v, u)) $
			\Else
				\State $ \enc{\firstCycle} \leftarrow \enc{\compGraphMatrix}(u, v) \sharedMult \enc{\compGraphMatrix}(v, w) \sharedMult \enc{\compGraphMatrix}(w, u) $
				\Statex \hspace*{4.4em} $ \sharedMult (\enc{\prioMatrix}(u, v) + \enc{\prioMatrix}(v, w) + \enc{\prioMatrix}(w, u)) $
				\vspace*{0.2em}
				\State $ \enc{\secondCycle} \leftarrow \enc{\compGraphMatrix}(u, w) \sharedMult \enc{\compGraphMatrix}(w, v) \sharedMult \enc{\compGraphMatrix}(v, u) $
				\Statex \hspace*{4.4em} $ \sharedMult (\enc{\prioMatrix}(u, w) + \enc{\prioMatrix}(w, v) + \enc{\prioMatrix}(v, u)) $
				\vspace*{0.2em}
				\State $ \enc{\chooseFirst}(i) \leftarrow \enc{\firstCycle} \sharedGreaterThanEqual \enc{\secondCycle} $
				\vspace*{0.1em}
				\State $ \enc{\cycleWeight}(i) \leftarrow \enc{\chooseFirst}(i) \ ? \ \enc{\firstCycle} \ : \ \enc{\secondCycle} $
			\EndIf
		\EndFor
		\Statex \underline{Approximation Phase}
		\vspace*{0.1em}
		\For{$ \lfloor \numparties / 2 \rfloor $}
			\State $ \enc{\indexHeaviest}, \enc{\heaviest} \leftarrow \searchHeaviestCycle(\enc{\cycleIndices}, \enc{\kepApNodes}, \enc{\cycleWeight}) $
			\vspace*{0.1em}
			\State $ \enc{\chosenSets} \leftarrow \enc{\chosenSets} + \demuxGate{\vert \subsets \vert}(\enc{\indexHeaviest}) $ 
			\vspace*{0.1em}
			\If{not last iteration}
				\For{$ 0 \leq i < 3 $ \textbf{in parallel}}
					\State $ \enc{\indicatorEmpty}(i) \leftarrow \demuxGate{\numparties}(\enc{\heaviest}(i)) $
				\EndFor
				\State $ \enc{\indicatorCombined} \leftarrow \enc{\indicatorEmpty}(0) + \enc{\indicatorEmpty}(1) + \enc{\indicatorEmpty}(2) $
				\For{$ 0 \leq i < \vert \subsets \vert $ \textbf{in parallel}}
					\State $ \enc{\cycleWeight}(i) \leftarrow \enc{\cycleWeight}(i) \sharedMult \prod_{v \in \kepApNodes(i)} (1 - \enc{\indicatorCombined}(v)) $ 
				\EndFor
				\vspace*{-0.2em}
			\EndIf
			\vspace*{-0.2em}
		\EndFor
		\Statex \underline{Resolution Phase}
		\For{$ 0 \leq i < \vert \subsets \vert $ \textbf{in parallel}}
			\State $ \enc{\subsetCycleMap}(i, 0) \leftarrow \enc{\chooseFirst}(i) \sharedMult \enc{\chosenSets}(i) $
			\State $ \enc{\subsetCycleMap}(i, 1) \leftarrow (\enc{1} - \enc{\chooseFirst}(i)) \cdot \enc{\chosenSets}(i) $
		\EndFor
		\For{$ 0 \leq i < \vert \subsets \vert $}
			\State $ u, v, w \leftarrow \kepApNodes(i) $
			\If{$ w = \numparties $}
				\State $ \enc{\solutionMatrix}(u, v) \leftarrow \enc{\solutionMatrix}(u, v) + \enc{\chosenSets}(i) $
				\State $ \enc{\solutionMatrix}(v, u)\hspace*{0.074em} \leftarrow \enc{\solutionMatrix}(v, u)\hspace*{0.074em} + \enc{\chosenSets}(i) $
			\Else
					\State $ \enc{\solutionMatrix}(u, v)\hspace*{0.1485em} \leftarrow \enc{\solutionMatrix}(u, v)\hspace*{0.1485em} + \enc{\subsetCycleMap}(i, 0) $
					\State $ \enc{\solutionMatrix}(v, w)\hspace*{0.055em} \leftarrow \enc{\solutionMatrix}(v, w)\hspace*{0.055em} + \enc{\subsetCycleMap}(i, 0) $
					\State $ \enc{\solutionMatrix}(w, u) \leftarrow \enc{\solutionMatrix}(w, u) + \enc{\subsetCycleMap}(i, 0) $
					\State $ \enc{\solutionMatrix}(u, w)\hspace*{-0.075em} \leftarrow \enc{\solutionMatrix}(u, w)\hspace*{-0.075em} + \enc{\subsetCycleMap}(i, 1) $
					\State $ \enc{\solutionMatrix}(w, v)\hspace*{0.055em} \leftarrow \enc{\solutionMatrix}(w, v)\hspace*{0.055em} + \enc{\subsetCycleMap}(i, 1) $
					\State $ \enc{\solutionMatrix}(v, u)\hspace*{0.2225em} \leftarrow \enc{\solutionMatrix}(v, u)\hspace*{0.2225em} + \enc{\subsetCycleMap}(i, 1) $
			\EndIf
			\vspace*{-0.125em}
		\EndFor 
		\State $ \enc{\solutionMatrix'} \leftarrow \reverseShuffleGate(\enc{\solutionMatrix}, \enc{\perm}) $
		\For{$ 1 \leq i \leq \numparties $ \textbf{in parallel}}
			\vspace*{0.1em}
			\State $ \enc{\outputDonors}(i) \leftarrow \sum_{j=1}^{\numparties}(j \sharedScalMult \enc{\solutionMatrix'}(j-1, i-1)) $
			\vspace*{0.2em}
			\State $ \enc{\outputRecipients}(i) \leftarrow \sum_{j=1}^{\numparties}(j \sharedScalMult \enc{\solutionMatrix'}(i-1, j-1)) $
		\EndFor
	\end{algorithmic}
\end{multicols}
\vspace*{-1.4em}
	\caption{$ \kepApProtocol $}
	\label{prot:kep_approx}
\end{protocolDouble}

\subsection{Protocol Specification}\label{sub:protocol_spec}
Protocol~\ref{prot:kep_approx} shows the specification of protocol $ \kepApProtocol $ which implements functionality~$ \kepApFunc $ (cf.~Functionality~\ref{func:kepAp}). The protocol consists of four phases which are described in the following. 

\paragraph{Construction Phase}
We first construct the compatibility graph based on the input quotes of the patient-donor pairs. To this end, we require a gate for determining whether the donor of a pair~$ \partysymbol{i} $ is medically compatible with the patient of another pair~$ \partysymbol{j} $ as well as a gate for computing the weight of the edge~$ (i, j) $ which is used for prioritization between different exchange cycles. 

We use the gate~$ \compCheck $ from~\cite{Breuer_KEprotocol_2020} for compatibility check based on the criteria blood type and HLA compatibility. This yields the matrix $ \enc{\compGraphMatrix} $ where $ \compGraphMatrix(i, j) = 1 $ iff the donor of~$ \partysymbol{i} $ and the patient of~$ \partysymbol{j} $ are medically compatible and $ \compGraphMatrix(i, j) = 0 $, otherwise. Thus, $ \compGraphMatrix(i, j) $ indicates if the edge $ (i, j) $ exists in the compatibility graph.

Since the criteria that are considered for prioritization vary widely across the existing kidney exchange platforms, we keep the gate~$ \prioGate $ generic and only provide a specific set of criteria for the run time evaluation (cf.~Section~\ref{sub:eval_runtime}). The gate~$ \prioGate $ is then used to compute the prioritization matrix~$ \enc{\prioMatrix} $ such that $ \enc{\prioMatrix}(i, j) $ encodes the weight of the edge $ (i, j) $ from the donor of pair~$ \partysymbol{i} $ to the patient of pair~$ \partysymbol{j} $. Afterwards, the matrices $ \enc{\compGraphMatrix} $ and $ \enc{\prioMatrix} $ are shuffled to guarantee for unbiasedness in the sense that the index of a patient-donor pair does not influence the probability that the pair gets matched. Then, we initialize the secret vector $ \enc{\chosenSets} $ to store for each subset whether it has been chosen or not, and the secret vector $ \enc{\cycleIndices} $ to store the index of each subset in~$ \subsets $.

\paragraph{Evaluation Phase}
We set the weight of each subset $ \singleSubset \in \subsets $ to the maximum weight of the exchange cycles that consist of the nodes from $ s $.
A subset $ s = \{u, v\} $ of size $ 2 $ is represented as a 3-tuple where the last entry $ w $ has the dummy value $ \numparties $. This allows us later on to hide whether a chosen subset is of size $ 2 $ or~$ 3 $. While for subsets of size $ 2 $, there is only one possible cycle, for each subset $ \{u, v, w\} $ of size $ 3 $, we have to evaluate the two cycles $ (u, v, w) $ and $ (u, w, v) $. We compute the weight of a cycle based on the matrices $ \enc{\compGraphMatrix} $ and $ \enc{\prioMatrix} $.
First, we use $ \enc{\compGraphMatrix} $ to check whether the cycle is executable, i.e., whether all its edges are present in the compatibility graph. To this end, we multiply the corresponding entries of $ \enc{\compGraphMatrix} $ with each other. The actual weight of the cycle is computed as the sum of the edge weights stored in the prioritization matrix $ \enc{\prioMatrix} $.
For each subset $ s $ of size $ 3 $, we choose the first cycle of maximum weight and set the weight of $ s $ to the weight of the chosen cycle. 
We then store for each subset, whether the first or the second cycle has been chosen, i.e., $ \enc{\chooseFirst}(i) = \enc{1} $ iff the first cycle has been chosen. Later on, this allows us to revert the mapping from cycles to subsets.

\paragraph{Approximation Phase}
We iteratively add the subset of maximum weight to the solution set~$ \enc{\chosenSets} $. At the beginning of each iteration, we use the gate $ \searchHeaviestCycle $ to obtain the index $ \enc{\indexHeaviest} $ and the nodes~$ \enc{\heaviest} $ of the subset of maximum weight. 
Then, we use the gate $ \demuxGate{\numparties} $ to obtain the binary indicator vector for the index of the chosen subset. This allows us to efficiently add the chosen subset to the solution set~$ \enc{\chosenSets} $. 

Unless we are in the last iteration, we have to update the subset weights afterwards such that the weight of all subsets that share a node with the chosen subset is set to $ \enc{0} $. To this end, we convert each node of the chosen subset into a binary indicator vector such that $ \enc{\indicatorEmpty}(i) $ is a vector whose $ j $-th entry is equal to $ \enc{1} $ iff $ \partysymbol{j} $ is the $ i $-th pair in the subset. All other entries of the vector are equal to~$ \enc{0} $. We then combine the three indicator vectors for the nodes of the chosen subset into a single indicator vector $ \enc{\indicatorCombined} $. 
To update the subset weights, we iterate over all subsets $ s = \{u, v, w\} \in \subsets $ and set the weight to $ \enc{0} $ if the indicator vector $ \enc{\indicatorCombined} $ is equal to $ \enc{1} $ for at least one of the nodes $ u, v, w $, i.e., if one of these nodes is part of the previously chosen subset. This can be efficiently achieved by multiplying the weight of the $ i $-th subset with $ \enc{1} - \enc{\indicatorCombined}(v) $ for each node in the subset. Note that this difference is equal to $ \enc{0} $ iff node $ v $ is part of the chosen subset. 
Recall that the subset returned by $ \searchHeaviestCycle $ is padded with the dummy node $ \vert \numparties \vert $ if it is empty or of size $ 2 $. In such a case, the last entry of $ \enc{\indicatorCombined} $ is equal to $ \enc{1} $. However, this does not influence the computation of the new weights as the
dummy node is not part of any subset in $ \subsets $.

Overall, we execute $ \lfloor \numparties / 2 \rfloor $ iterations of the main loop. This corresponds to the worst case as in each iteration we choose one~subset of size at least $ 2 $ as long as there are subsets of weight larger than $ 0 $. Our protocol becomes entirely data oblivious by executing dummy iterations if there are no more subsets of weight larger than~0.

\paragraph{Resolution Phase}
We compute the secret mapping from subsets of size $ 3 $ to cycles of size $ 3 $ such that $ \enc{\subsetCycleMap}(i, j) $ indicates whether the $ j $-th cycle of the $ i $-th subset is part of the computed solution. 
Recall that $ \enc{\chooseFirst}(i) = \enc{1} $ iff the first cycle has been chosen. Thus, we can multiply $ \enc{\chooseFirst}(i) $ with $ \enc{\chosenSets}(i) $ to obtain $ \enc{\subsetCycleMap}(i, 0) $. 
Similarly, we compute $ \enc{\subsetCycleMap}(i, 1) $ using $ \enc{1} - \enc{\chooseFirst}(i) $.

Then, we compute the solution matrix $ \solutionMatrix $ where $ \enc{\solutionMatrix}(u, v) $ states that the donor of the $ u $-th patient-donor pair donates her kidney to the patient of the $ v $-th pair. For subsets of size $ 2 $, we add $ \enc{\chosenCycles}(i) $ to those entries of $ \enc{\solutionMatrix} $ that correspond to the edges of the $ i $-th cycle. Thus, if the cycle is part of the solution, these entries are set to $ \enc{1} $. For the subsets of size $ 3 $, we use the subset map $ \enc{\subsetCycleMap} $ in a similar fashion to determine for all edges in the two cycles $ (u, v, w) $ and $ (u, w, v) $ if they are in the~solution. 

Afterwards, we revert the initial shuffling of the compatibility graph by applying the gate $ \reverseShuffleGate $ to the solution matrix $ \enc{\solutionMatrix} $. Finally, we iterate over each row $ i $ in $ \enc{\solutionMatrix} $ and sum up the product of each entry $ \enc{\solutionMatrix}(i-1, j-1) $ and the public index $ j $. Since each row in $ \enc{\solutionMatrix} $ contains at most one entry that is equal to $ 1 $, this sum yields the index of the recipient for the donor of pair $ \partysymbol{i} $. Similarly, we use the columns of $ \enc{\solutionMatrix} $ to determine the donor for each pair.

\paragraph{Security}
As protocol~$ \kepApProtocol $ is data oblivious and only calls secure gates, it can be simulated by applying the composition theorem from~\cite{Canetti_UC_2001} and following the known flow of execution.

\paragraph{Complexity}
The protocol~$ \kepApProtocol $ includes $ \lfloor \numparties / 2 \rfloor $ calls to the gate $ \searchHeaviestCycle $, which requires $ \bigo{\numparties^3 \bitsize} $ communication and $ \bigo{\log \numparties^3 \log \bitsize} $ rounds for $ \numparties $ patient-donor pairs and maximum cycle size $ 3 $. Thus, overall the protocol $ \kepApProtocol $ has communication complexity $ \bigo{\numparties^4 \bitsize} $ and round complexity $ \bigo{\numparties \log \numparties^3 \log \bitsize} $. 
\section{Evaluation}\label{sec:evaluation}
We evaluate our protocol $ \kepApProtocol $ in terms of run time performance and approximation quality using a real-world data set from the United Network for Organ Sharing (UNOS).
The data set comprises the data of 2,913 unique patient-donor pairs that registered with UNOS between 27th October 2010 and 29th December 2020.

\subsection{Evaluation of the Run Time Performance}\label{sub:eval_runtime}

\paragraph{Setup}
We implement our protocol in the SMPC benchmarking framework MP-SPDZ~\cite{KellerMPSPDZ2020} version 0.3.5. We consider the same setup as~\cite{Breuer_Matching_2022,Breuer_KepIp_2022}, where each computing peer runs on an LXC container with Ubuntu 20.04 LTS, one core, and 32 GB RAM. One extra container is used for managing input and output for the patient-donor pairs. All containers run on a single server. 

While we kept the criteria for prioritization generic in the specification of the~protocol $ \kepApProtocol $ (cf.~Section~\ref{sub:protocol_spec}), for the evaluation we have to decide~on~a specific set of criteria. In order to use realistic~values for the input of the patient-donor pairs, we consider all criteria~that can be derived from our real-world data set from UNOS, i.e., the patient's age, calculated panel reactive antibody, pediatric status, and prior living donor status, as well as common HLA or blood type of patient and donor, age difference of patient and donor, and their region (i.e., geographic proximity). 

We evaluate our protocol for three different SMPC primitives: (1) \texttt{rep-ring}~\cite{Araki_ReplicatedSecretSharing_2016}, which provides for security in the semi-honest model with honest majority, (2) \texttt{ps-rep-ring}~\cite{Eerikson_UseYourBrain_2020}, which provides for security in the malicious model with honest majority, (3) \texttt{semi2k}, which is an adaptation of~\cite{Cramer_Spdz2k_2018} providing for security in the semi-honest model with dishonest majority. The latter thus considers the same security model as~\cite{Birka_PPKidneyExchange_2022}, which allows for a fair comparison between~\cite{Birka_PPKidneyExchange_2022} and our protocol~$ \kepApProtocol $. We use three computing peers for the honest majority case and two for the dishonest majority case. More details on the SMPC primitives are provided in Appendix~\ref{app:smpc_protocols}.

For the network setting, we follow the related work~\cite{Breuer_Matching_2022,Breuer_KepIp_2022,Birka_PPKidneyExchange_2022}~and consider a LAN setting with a bandwidth of $ 1 $Gbps and a latency of $ 1 $ms. 
We also evaluate our protocol in a more realistic scenario which we call the WAN setting. We still deem a bandwidth of $ 1 $Gbps realistic as the computing peers are hosted by large institutions with high-bandwidth internet connection. However, with increased distance, the latency between the computing peers increases. In the WAN setting, we use a latency of $ 20 $ms which can be considered a worst case scenario as the computing peers are hosted by institutions such as universities or large hospitals. We emulate these network settings between the containers using \texttt{tc} (traffic control).

We measure run time and network traffic for different numbers of patient-donor pairs. For each number of pairs, we run $ 10 $ repetitions if the run time is below $ 1 $ hour, and $ 1 $ repetition, otherwise. This suffices as our protocol is data oblivious and thus there are small differences between the run times for multiple repetitions. For example, for $ 100 $ pairs (\texttt{rep-ring}) in the LAN setting, the run times for all $ 10 $ repetitions lie in a range of $ 1 $ second.
For the patient-donor pairs' input, we sample uniformly at random from the UNOS data set. Note, however, that due to the data obliviousness of our protocol, run time and network traffic are independent from~the~input.

\begin{figure}
	\includegraphics*[scale=1.0]{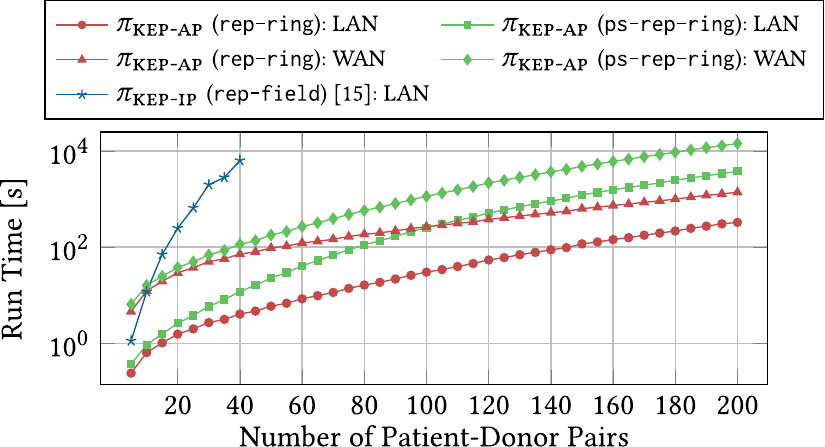}
	\caption{Run time for our protocol~$ \kepApProtocol $ and the protocol $ \kepIpProtocol $~\cite{Breuer_KepIp_2022,Breuer_KepIpExtended_2022} that computes an exact solution for the KEP.
	}
	\label{plt:runtimes}
\end{figure}

\paragraph{Results}
Figure~\ref{plt:runtimes} shows the run time for our protocol~$ \kepApProtocol $ with maximum cycle size $ 3 $ for \texttt{rep-ring} and \texttt{ps-rep-ring}, i.e., considering an honest majority. We provide the evaluation for cycle size $ 2 $ in Appendix~\ref{app:twocycles} and the raw data for Figure~\ref{plt:runtimes} in Appendix~\ref{app:results}. 

We observe that the run time of our protocol~$ \kepApProtocol $ scales sub-exponentially in the number of patient-donor pairs. This goes in line with its theoretical complexities which we determined in Section~\ref{sub:protocol_spec}. 
As expected, the protocol run times are faster for the semi-honest model than for the malicious model. On average the protocol takes about $ 6.2 $ times longer in the malicious model.
When comparing the run times for the different network settings, we observe that the run times in the WAN setting are about $ 10.4 $ times longer than in the LAN setting for the semi-honest model and about $ 6.5 $ times longer for the malicious model. While this is a non-negligible performance overhead, our protocol still runs for less than $ 4 $ hours for $ 200 $ patient-donor pairs in the WAN setting for the malicious model. Considering that in case of existing kidney exchange platforms, the KEP is solved at most once per day~\cite{AshlagiKidneyExchangeOperations2021,Biro_EuropeanModellingKE_2019}, such run times are acceptable in practice. Note that even in a large platform such as UNOS a match run only involves about $ 200 $ pairs.

\paragraph{Comparison to Related Work}
Figure~\ref{plt:runtimes} also shows the run time of the protocol~$ \kepIpProtocol $ from~\cite{Breuer_KepIpExtended_2022,Breuer_KepIp_2022}, which is the most efficient SMPC protocol to date for computing an exact solution for the KEP. The protocol~$ \kepIpProtocol $ is evaluated in~\cite{Breuer_KepIpExtended_2022} using the SMPC primitive~\texttt{rep-field} from~\cite{Araki_ReplicatedSecretSharing_2016} (i.e., \texttt{rep-ring} over $ \field_p $ instead of $ \ints_{2^k} $). We observe that our protocol~$ \kepApProtocol $ achieves superior run times compared to the protocol~$ \kepIpProtocol $. For example, in the LAN setting with the semi-honest model our protocol finishes within $ 4 $~seconds for $ 40 $ pairs whereas the protocol~$ \kepIpProtocol $ requires more than $ 105 $ minutes. This is, however, to be expected as computing an optimal solution for the KEP is an NP-complete problem~\cite{AbrahamClearingAlgorithms2007} and thus our protocol solves a much easier problem than the protocol~$ \kepIpProtocol $.

Figure~\ref{plt:birka_comparison} shows the run time of our protocol~$ \kepApProtocol $ with a cycle size up to $ 3 $ in the semi-honest model with dishonest majority compared to the protocol from~\cite{Birka_PPKidneyExchange_2022} with a fixed cycle size of $ 3 $.\footnote{In addition, in Appendix~\ref{app:twocycles} we compare the two protocols for a cycle size of $ 2 $.} 
Since~\cite{Birka_PPKidneyExchange_2022} use the ABY framework~\cite{Demmler_ABY_2015}, which considers two computing peers, they also consider the semi-honest model with dishonest majority.
We observe that our protocol~$ \kepApProtocol $ achieves significantly faster run times for $ 13 $ or more patient-donor pairs, already outperforming the run times from~\cite{Birka_PPKidneyExchange_2022} by a factor of $ 5.4 $ for $ 15 $ pairs. 
Due to the high RAM consumption, \cite{Birka_PPKidneyExchange_2022} extrapolate their run times for more than $ 18 $ pairs. 
The extrapolated run time for $ 40 $ pairs is more than $ 24 $ hours whereas our protocol~$ \kepApProtocol $ finishes within $ 2 $ minutes for $ 40 $ pairs. 
While the specific run times are not directly comparable due to the implementation in different frameworks, our evaluation still demonstrates that our protocol outperforms the protocol from~\cite{Birka_PPKidneyExchange_2022} by orders of magnitude and that~\cite{Birka_PPKidneyExchange_2022} has a significantly steeper curve.
This holds irrespective of the specific framework. 
We attribute this to the entirely different algorithmic approaches of the two protocols (cf.~Section~\ref{sub:rw_ppKE}). 

\begin{figure}[t]
	\includegraphics*[scale=1.0]{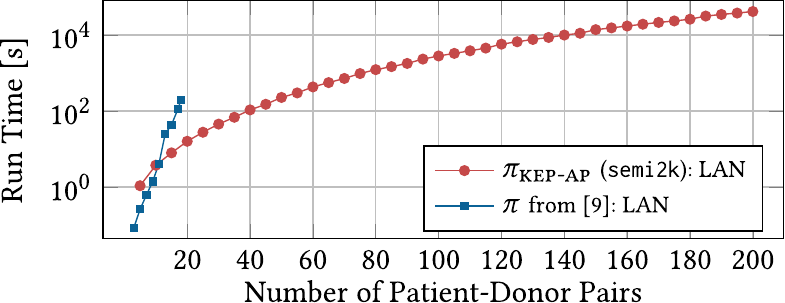}
	\caption{Run time for our protocol~$ \kepApProtocol $ with \texttt{semi2k} and the results from~\cite{Birka_PPKidneyExchange_2022} which uses the ABY framework~\cite{Demmler_ABY_2015}. The latency in the LAN setting from~\cite{Birka_PPKidneyExchange_2022} is~1.3ms.}
	\label{plt:birka_comparison}
\end{figure}

\paragraph{Further Improvements}
An alternative to the three underlying SMPC primitives used in our evaluation is function secret sharing (FSS). The idea of FSS is to additively secret-share a function, which can then be evaluated locally~\cite{Boyle_Fss_2015}. FSS has been shown to achieve extremely fast online run times (e.g., comparisons only require constant rounds and sub-linear communication~\cite{Boyle_FssMixed_2021}). On the other hand FSS requires the generation of a huge amount of keying material, increasing offline communication and online storage~\cite{Boyle_FssMixed_2021}. Still, recent research (especially, in private machine learning~\cite{Jawalker_FssPPML_2023,Wagh_FssPPML_2022,Storrier_Grotto_2023,Gupta_Llama_2022}) has shown a significant performance gain achieved by FSS, sometimes yielding run times up to 693 times faster than for SMPC primitives such as \texttt{rep-ring}~\cite{Jawalker_FssPPML_2023}. While such performance gains cannot be directly translated to our use case, the superior round complexity achieved by FSS suggests that the run time of our protocol~$ \kepApProtocol $ could be further improved using FSS. However, as MP-SPDZ does not support FSS, we leave the evaluation with FSS for future work.
 
\subsection{Evaluation of the Approximation Quality}\label{sub:eval_quality}

\paragraph{Setup}
We evaluate the quality of the approximation obtained for our protocol~$ \kepApProtocol $ as well as for the protocol from~\cite{Birka_PPKidneyExchange_2022}. 
To this end, we implemented both protocols as a conventional algorithm which does not employ any SMPC techniques. For our protocol this corresponds to implementing Algorithm~\ref{alg:greedy}.
The results obtained by these conventional algorithms are exactly the same as for the SMPC protocols. The advantage of the implementation as a conventional algorithm is that we can compute the approximation efficiently for a large number of repetitions. Again, we sample the inputs of the patient-donor pairs uniformly at random from the UNOS data set. For each run, we determine the result obtained by the conventional algorithm as well as the optimal result obtained by an LP solver for the KEP. 
We executed 100 iterations for each number of pairs and computed the approximation quality as the percentage of pairs that are matched in the approximation approach compared to the optimal solution obtained by the LP solver. As the criteria considered for prioritization vary widely across existing kidney exchange platforms~\cite{AshlagiKidneyExchangeOperations2021,Biro_EuropeanModellingKE_2019}, it is not trivial to determine a representative prioritization function. The only criterion that is commonly accepted as the most important criterion is to maximize the overall number of transplants~\cite{AshlagiKidneyExchangeOperations2021,Biro_EuropeanModellingKE_2019}. Therefore, we use this as the only criterion for the evaluation of the approximation quality. Note that the obtained quality may vary slightly for other prioritization functions.

\paragraph{Results}

\begin{figure}
	\includegraphics*[scale=1.0]{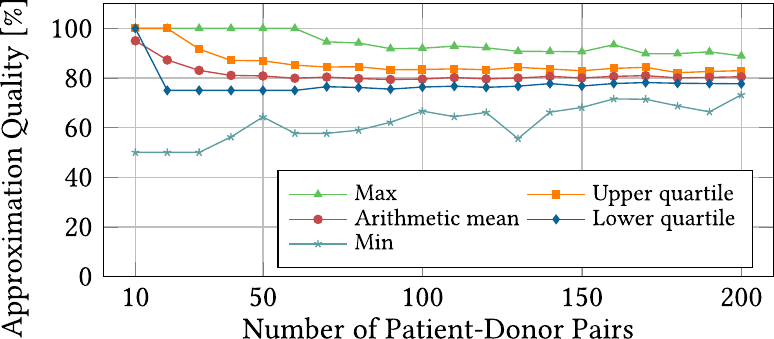}
	\caption{Quality of the solution obtained by our protocol~$ \kepApProtocol $ compared to an optimal solution for the KEP.}
	\label{plt:offline_eval}
\end{figure}

\begin{figure}
	\includegraphics*[scale=1.0]{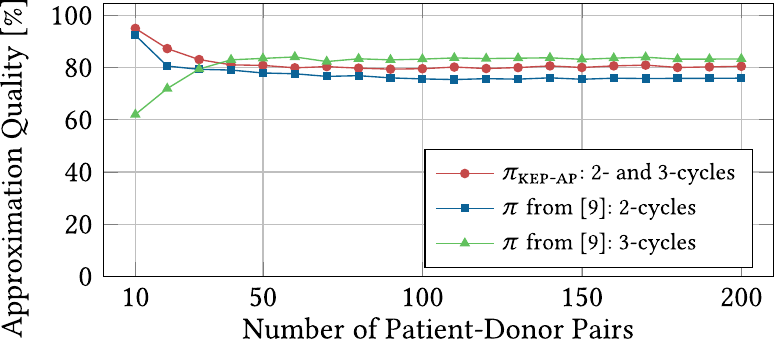}
	\caption{Arithmetic mean of the approximation quality for our protocol $ \kepApProtocol $ compared to the protocol from~\cite{Birka_PPKidneyExchange_2022}.}
	\label{plt:offline_eval_comparison}
\end{figure}

Figure~\ref{plt:offline_eval} shows the results of the evaluation of the approximation quality for our protocol~$ \kepApProtocol $. We observe that the approximation quality is always above the theoretical approximation ratio of $ 1/3 $. In fact, in our tests the quality is never below $ 50\% $ and on average we always match around $ 80\% $ of the possible pairs if the protocol is run for more than $ 50 $ pairs. If the number of pairs is lower, the quality of our approximation is even better. E.g., for $ 10 $ pairs, the quality of our approximation is $ 95\% $ on average.

Figure~\ref{plt:offline_eval_comparison} shows the average approximation quality of our protocol~$ \kepApProtocol $ compared to the quality obtained by~\cite{Birka_PPKidneyExchange_2022}, once for $ 2 $-cycles only and once for $ 3 $-cycles only. 
We observe that our protocol achieves a better quality for all numbers of patient-donor pairs compared to~\cite{Birka_PPKidneyExchange_2022} with $ 2 $-cycles only. This goes in line with the common result that $ 3 $-cycles increase the number of pairs that can be matched~\cite{RothEfficientKidneyExchange2007}.
Compared to~\cite{Birka_PPKidneyExchange_2022} for $ 3 $-cycles only, our protocol achieves a better quality for less than $ 40 $~pairs. For 40 pairs and more, the approach from~\cite{Birka_PPKidneyExchange_2022}, which evaluates one solution for each cycle that exists in the compatibility graph and then chooses the best among these (cf.~Section~\ref{sub:rw_ppKE}), seems to outweigh the benefit of considering both $ 2 $- and $ 3 $-cycles.
However, the gap between the quality obtained by \cite{Birka_PPKidneyExchange_2022} and our protocol stays at about $ 3.9 $\%. 
This slightly superior approximation quality of~\cite{Birka_PPKidneyExchange_2022} comes at the cost of data obliviousness as in \cite{Birka_PPKidneyExchange_2022} the flow of execution depends on the number of cycles in the private compatibility graph. 
Also recall that \cite{Birka_PPKidneyExchange_2022} does not scale for the large numbers of pairs that are to be expected in practice (e.g., about $ 200 $ pairs for UNOS). Thus, the slightly superior approximation quality of~\cite{Birka_PPKidneyExchange_2022} for larger numbers of pairs does not yield an advantage in practice over our protocol.

\section{Simulation}\label{sec:simulation}

In Section~\ref{sec:evaluation}, we evaluated our protocol~$ \kepApProtocol $ in a static setting where a fixed set of patient-donor pairs send their input to the computing peers once. However, in practice kidney exchange is a dynamic problem since the KEP has to be solved repeatedly over time on varying sets of patient-donor pairs. In this section, we evaluate our protocol $ \kepApProtocol $ when used in the context of a dynamic kidney exchange platform where patient-donor pairs register and de-register over time. 
This allows us to determine the impact of using our protocol~$ \kepApProtocol $ in practice. In particular, we compare the number of transplants over time obtained when using our approximation protocol~$ \kepApProtocol $ for each match run compared to using a conventional IP solver that computes an exact solution for each run.
To this end, we use the same simulation framework used by~\cite{Breuer_Matching_2022} for evaluating their protocol for crossover exchange in a dynamic setting. The framework is based on DESMO-J~\cite{DESMOJ}.

\begin{figure}[t]
	\includegraphics*[scale=1.0]{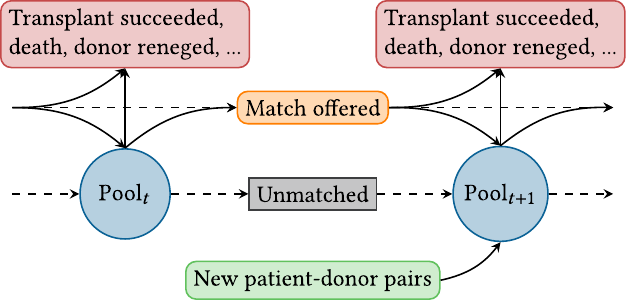}
	\caption{Dynamics in a kidney exchange platform~\cite{Breuer_Matching_2022,Dickerson_FailureAwareKE_2019}.}
	\label{fig:dynamic_model}
\end{figure}

Figure~\ref{fig:dynamic_model} shows the generic model of a dynamic kidney exchange platform. At the core of such a platform is the pool of registered patient-donor pairs. The average number of days it takes for a new pair to register is called the \emph{arrival rate}. The KEP is then solved for all pairs in the pool at fixed time intervals. This is also referred to as a \emph{match run} and the time span between two match runs is the \emph{match run interval}. The specific models for which we run simulations then differ in how they solve the KEP at each match run. 
After a match run, those pairs for which a match has been found have to undergo further medical examination called a \emph{crossmatch test}. The probability that this test fails is called the \emph{crossmatch failure probability}. This probability is about $ 35\% $ for highly sensitized patients and $ 10\% $ for others~\cite{AshlagiKidneyExchangeOperations2021}. If the pairs pass the crossmatch test, they are offered the matching organ and can either accept it or refuse it. The probability that a patient refuses the matching organ is called the \emph{match refusal rate}. If the crossmatch fails or a pair refuses the matching organ, all pairs involved in the exchange reenter the pool. On average it takes a pair about $ 7 $ days to reenter in case of a failed crossmatch and about $ 2 $ days in case of a match refusal~\cite{AshlagiFrequencies2018}. Besides leaving the platform due to a found match, pairs may also leave unmatched due to death, a reneged donor, or other reasons. The average number of days a pair stays in the pool before leaving due to any other reason than being matched is called the \emph{departure~rate}.

Table~\ref{tab:sim_parameters} lists the parameters that exist in a kidney exchange platform and the values of these parameters for which we run simulations. Note that the existing platforms vary widely w.r.t.\ these values~\cite{AshlagiFrequencies2018,AshlagiKidneyExchangeOperations2021,Biro_EuropeanModellingKE_2019}. To evaluate the performance of a platform using our protocol~$ \kepApProtocol $ for the different settings encountered in practice, we follow \cite{Breuer_Matching_2022} and run simulations for all values listed~in~Table~\ref{tab:sim_parameters}. 

\subsection{Setup}\label{sub:sim_setup}

We use the same data set from UNOS as in the run time evaluation. For each simulation run, we sample the involved patient-donor pairs uniformly at random from the set of 2,913 unique pairs contained in the UNOS data set and we consider a time horizon of $ 5 $ years. We then evaluate the number of patients that receive a kidney transplant over time. We execute $ 50 $ repetitions for each parameter constellation and average the results. 

\begin{table}[t]
	\caption{Parameters of a kidney exchange platform and the values considered in our simulations (adopted from~\cite{Breuer_Matching_2022}).}
	\label{tab:sim_parameters}
	\begin{adjustbox}{width=\columnwidth}
\ra{0.9}
\begin{tabular}{l | l}
	Parameter 						& Values 									\\\hline
	arrival rate					& $ 1, 2, 4, 7, 14 $ days 					\Tstrut\\
	match run interval				& $ 1, 2, 4, 7, 14, 30, 60, 120 $ days	\\
	departure rate					& $ 400, 800, 1200 $ days					\\
	match refusal rate				& $ 0, 10, 20, 30, 40 $ \%					\\
	crossmatch failure prob.\		& $ 35 $ \%, for highly sensitized patients	\\
									& $ 10 $ \%, for others						\\
	reentering delay				& $ 7 $ days, for a failed crossmatch		\\
									& $ 2 $ days, for a match refusal			\\
\end{tabular}
\end{adjustbox}
\end{table}

For our simulations, we consider three different models. They all follow the generic model depicted in Figure~\ref{fig:dynamic_model} but use different approaches for solving the KEP at each match run. In the \emph{conventional model} the KEP is solved by a central platform operator using a conventional IP solver. The privacy-preserving \emph{$ \kepApProtocol $ model} uses our novel protocol~$ \kepApProtocol $ to approximate the KEP in each match run. The privacy-preserving \emph{$ \kepIpProtocol $ model} uses the protocol~$ \kepIpProtocol $ from~\cite{Breuer_KepIp_2022}, which is the best known SMPC protocol for optimally solving the KEP. As for our qualiy evaluation (cf.~Section~\ref{sub:eval_quality}), we consider the prioritization function that maximizes the number of found transplants, which is universally considered as the most important criterion in kidney exchange~\cite{AshlagiKidneyExchangeOperations2021,Biro_EuropeanModellingKE_2019}. All other prioritization criteria vary widely across the different platforms and including any of them would bias the results towards those specific criteria.

In the conventional model, we use a standard LP solver for computing the optimal solution for the KEP whereas in the $ \kepApProtocol $ model we use the conventional algorithm (Algorithm~\ref{alg:greedy}) from the quality evaluation (Section~\ref{sub:eval_quality}) to determine the number of matched patients. Then, we use the run times from Section~\ref{sub:eval_runtime} and advance the simulation time by this amount. Thereby, we can run simulations for a large time horizon and many different parameter constellations in a short time. The simulation results presented in this section can be achieved using any underlying SMPC primitive for our protocol $ \kepApProtocol $ as long as the run time is below 24 hours for 200 patient-donor pairs. This holds for all settings from Section~\ref{sub:eval_runtime}. If a protocol is used that yields run times longer than 24 hours for 200 pairs, this can still be simulated but would yield different~results.

For the $ \kepIpProtocol $ model, we use the same LP solver as in the conventional model and set the simulation time according to the run~times from~\cite{Breuer_KepIp_2022}. Note that the protocol~$ \kepIpProtocol $ only exhibits feasible run times for small numbers of pairs. If the run time is not feasible for the number of pairs in the pool of the kidney exchange platform, we split the pool into multiple sub-pools and solve the KEP for each of these independently. In particular, we split the pool into the minimum number of sub-pools for which the run time of the protocol $ \kepIpProtocol $ is still feasible. This yields an approximation of the overall optimal solution as possibly matching pairs may be assigned to different sub-pools. 
Thus, in the $ \kepApProtocol $ model we compute an approximation on the whole set of pairs in the pool whereas in the $ \kepIpProtocol $ model we obtain an approximation due to pool splitting. 

\subsection{Results}\label{sub:sim_results}

\begin{figure}[t]
	\includegraphics*[scale=0.96]{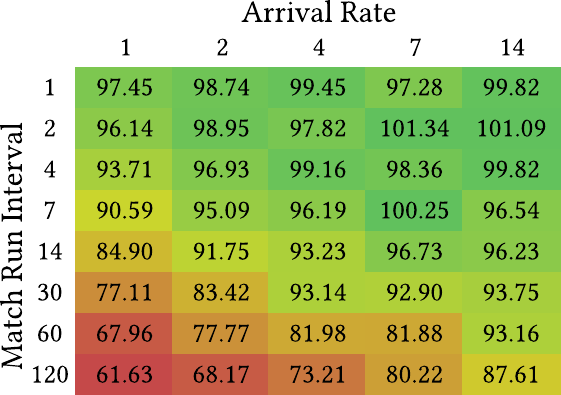}
	\caption{$ \kepApProtocol $ model vs.\ conventional model - the percentages state the number of transplants achieved in the $ \kepApProtocol $ model divided by those in the conventional model for a departure rate of 400 days and a match refusal rate~of~20\%.}
	\label{tab:heatmap_cv_ap}
\end{figure}

\paragraph{$ \kepApProtocol $ Model vs.\ Conventional Model}
We compare the simulation results for the $ \kepApProtocol $ model to those of the conventional model by computing the percentage of transplants achieved in the $ \kepApProtocol $ model compared to those achieved in the conventional model. Our results show that the performance difference between the two models is heavily influenced by the arrival rate and the match run interval but only very slightly by the departure rate and the match refusal rate. In particular, the overall percentage of transplants in the $ \kepApProtocol $ model compared to the conventional model is within $ 1\% $ for all departure rates and within $ 4\% $ for all match refusal rates listed in Table~\ref{tab:sim_parameters}, slightly increasing for larger departure rates and slightly decreasing for larger match refusal rates. As the influence of departure and match refusal rate on the performance difference is small, we only present the results for fixed values of these parameters. Thus, we focus on those parameters that significantly influence the performance difference between the different~models.

Figure~\ref{tab:heatmap_cv_ap} shows the percentages for a fixed departure rate of $ 400 $ days, a fixed match refusal rate of $ 20\% $, and different values for the arrival rate and the match run interval. Larger percentages indicate a better performance of the $ \kepApProtocol $ model compared to the conventional model. We make three main observations.

First, the $ \kepApProtocol $ model performs considerably better than suggested by the quality evaluation in the static setting, where the approximation quality is about $ 80\% $ (cf.~Section~\ref{sub:eval_quality}). This is due to the fact that a pair that is not matched in one match run can still be matched in a subsequent run. More specifically, the optimal solution computed for a single run in the conventional model is not necessarily the optimal solution across multiple runs. We show this on a small, simplified example. Assume that in a match run with the four pairs $ \partysymbol{A} $, $ \partysymbol{B} $, $ \partysymbol{C} $, and $ \partysymbol{D} $, the conventional model finds the cycles $ (\partysymbol{A}, \partysymbol{B}) $ and $ (\partysymbol{C}, \partysymbol{D}) $ whereas the $ \kepApProtocol $ model only finds the cycle $ (\partysymbol{B}, \partysymbol{C}) $. Considering this single run, the $ \kepApProtocol $ model only finds $ 50\% $ of the transplants found in the conventional model. Assume now that in the next run, a pair $ \partysymbol{E} $ has joined the pool such that it forms the 3-cycle $ (\partysymbol{A}, \partysymbol{D}, \partysymbol{E}) $  with $ \partysymbol{A} $ and $ \partysymbol{D} $. Since in the conventional model these pairs already left the pool in the previous run, this cycle cannot be found. However, in the $ \kepApProtocol $ model, the cycle is found. Thus, in this simplified example, the $ \kepApProtocol $ model even matches more pairs over two runs than the conventional model. These dynamics lead to the $ \kepApProtocol $ model performing much better than indicated by the quality evaluation in the static setting. For some parameter constellations, this even leads to a better performance of the $ \kepApProtocol $ model compared to the conventional model (e.g., all percentages larger than $ 100\% $ in~Figure~\ref{tab:heatmap_cv_ap}).

Second, the $ \kepApProtocol $ model performs better for smaller match run intervals. If match runs are executed frequently, the $ \kepApProtocol $~model has a higher probability of matching pairs in later match runs that could not be matched in a previous run as the probability that a pair leaves in between two runs due to any other reason than being matched is smaller for small match run intervals. Also, the average pool size is smaller for small match run intervals as less pairs arrive between two runs. The $ \kepApProtocol $ model benefits from such smaller pool sizes as the approximation computed by the protocol $ \kepApProtocol $ is better for smaller numbers of patient-donor pairs (cf.~Section~\ref{sub:eval_quality}).

Third, the $ \kepApProtocol $ model performs better for larger arrival rates. Since for larger arrival rates overall less pairs register with the platform, the average pool size over time is also less. As stated above, the $ \kepApProtocol $ model benefits from this as it achieves a better approximation quality for smaller numbers of pairs (cf.~Section~\ref{sub:eval_quality}).
 
Overall, our results show that the $ \kepApProtocol $ model achieves nearly the same number of transplants over time if the match run interval is smaller than $ 7 $ days or if it is smaller than $ 14 $ days and at the same time the arrival rate is larger than $ 2 $ days. Thus, it only performs considerably worse than the conventional model if pairs register frequently and match runs are only executed very infrequently. This, however, also corresponds to the parameter constellations that yield the lowest numbers of transplants. As the match run interval can be freely chosen by a kidney exchange platform, in practice it would not make sense to choose a large match run interval if pairs register very frequently. Recent research comes to the same conclusion stating that in general small match run intervals yield a larger number of transplants~\cite{AshlagiFrequencies2018}. Thus, we conclude that the performance difference between our $ \kepApProtocol $ model and the conventional model is very small for those parameter constellations that occur in practice.

\begin{figure}[t]
	\includegraphics*[scale=0.96]{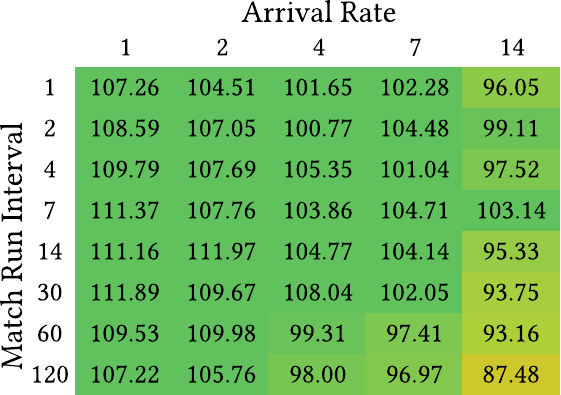}
	\caption{$ \kepApProtocol $ model vs.\ $ \kepIpProtocol $ model - the percentages state the number of transplants in the $ \kepApProtocol $ model divided by those in the $ \kepIpProtocol $ model with pool splitting both for a departure rate of 400 days and a match refusal rate of 20\%.}
	\label{tab:heatmap_ip_ap}
\end{figure}

While one may argue that even a high value such as $ 97\% $ of transplants compared to the conventional model is unacceptable in practice, this only holds for those countries where kidney exchange is already possible. In countries such as Germany or Brazil, however, where kidney exchange is legally not allowed due to fear of manipulation, corruption, and coercion, our privacy-preserving approach has the potential to increase the number of transplants from $ 0 $ to $ 97\% $ of the maximum possible number of transplants.

\paragraph{$ \kepApProtocol $ Model vs.\ $ \kepIpProtocol $ Model}
Figure~\ref{tab:heatmap_ip_ap} shows the results for the comparison of our $ \kepApProtocol $ model and the $ \kepIpProtocol $ model which uses the protocol~$ \kepIpProtocol $ from~\cite{Breuer_KepIp_2022}. 
Note that values larger than $ 100 $\% indicate that our $ \kepApProtocol $ model outperforms the $ \kepIpProtocol $ model. Recall that the protocol~$ \kepIpProtocol $ computes an exact solution for the KEP but only scales for small numbers of patient-donor pairs. If the number of pairs becomes too large for the protocol~$ \kepIpProtocol $, we split the pool into multiple sub-pools and then solve the KEP separately for each sub-pool. Thus, the $ \kepIpProtocol $ model also only computes an approximation of the KEP if the pool size becomes too large. 

The percentages shown in Figure~\ref{tab:heatmap_ip_ap} show that our $ \kepApProtocol $ model outperforms the $ \kepIpProtocol $ model for nearly all parameter constellations. 
In general, the performance of the $ \kepIpProtocol $ model is closer to the $ \kepApProtocol $ model for larger arrival rates as these induce a smaller average pool size and thus the pool has to be split into less sub-pools.
Therefore, the approximation obtained due to pool splitting is closer to the approximation obtained by our protocol $ \kepApProtocol $.

\section{Conclusion and Future Work}

We have presented an SMPC protocol that approximates the KEP using a Greedy approach and evaluated it both for the semi-honest and the malicious model. In contrast to the state-of-the-art protocols for the KEP, our protocol is data oblivious and scales for the large numbers of patient-donor pairs encountered in existing kidney exchange platforms. We have shown that our protocol outperforms the existing exact privacy-preserving solutions for the KEP in the dynamic setting, where patient-donor pairs register and de-register over time and that the difference in the number of transplants found compared to the conventional model is small in practice.

A possible direction for future work is the development of an~approximation protocol that allows for altruistic donors who are not bound to any particular patient. These are not allowed in all countries but have been shown to increase the number of transplants.

\section*{Source Code}
The MP-SPDZ source code of our protocol~$ \kepApProtocol $ (cf.~Protocol~\ref{prot:kep_approx}) is available under:
\url{https://gitlab.com/rwth-itsec/kep-ap}

\begin{acks}
This work was funded by the Deutsche Forschungsgemeinschaft (DFG, German Resarch Foundation) - project number (419340256) and NSF grant CCF-1646999. Any opinion, findings, and conclusions or recommendations expressed in this material are those of the author(s) and do not necessarily reflect the views of the National Science Foundation. Simulations were performed with computing resources granted by RWTH Aachen University under project rwth1214.
\end{acks}

\bibliographystyle{ACM-Reference-Format}
\bibliography{references_doi}


\begin{thebibliography}{42}


\ifx \showCODEN    \undefined \def \showCODEN     #1{\unskip}     \fi
\ifx \showDOI      \undefined \def \showDOI       #1{#1}\fi
\ifx \showISBNx    \undefined \def \showISBNx     #1{\unskip}     \fi
\ifx \showISBNxiii \undefined \def \showISBNxiii  #1{\unskip}     \fi
\ifx \showISSN     \undefined \def \showISSN      #1{\unskip}     \fi
\ifx \showLCCN     \undefined \def \showLCCN      #1{\unskip}     \fi
\ifx \shownote     \undefined \def \shownote      #1{#1}          \fi
\ifx \showarticletitle \undefined \def \showarticletitle #1{#1}   \fi
\ifx \showURL      \undefined \def \showURL       {\relax}        \fi
\providecommand\bibfield[2]{#2}
\providecommand\bibinfo[2]{#2}
\providecommand\natexlab[1]{#1}
\providecommand\showeprint[2][]{arXiv:#2}

\bibitem[Abraham et~al\mbox{.}(2007)]%
        {AbrahamClearingAlgorithms2007}
\bibfield{author}{\bibinfo{person}{David~J Abraham}, \bibinfo{person}{Avrim
  Blum}, {and} \bibinfo{person}{Tuomas Sandholm}.}
  \bibinfo{year}{2007}\natexlab{}.
\newblock \showarticletitle{Clearing Algorithms for Barter Exchange Markets:
  Enabling Nationwide Kidney Exchange}. In \bibinfo{booktitle}{\emph{ACM
  Conference on Electronic Commerce}}. \bibinfo{publisher}{ACM}.
\newblock
\newblock
\shownote{\url{https://doi.org/10.1145/1250910.1250954}}.


\bibitem[Anderson et~al\mbox{.}(2015)]%
        {AndersonTSP2015}
\bibfield{author}{\bibinfo{person}{Ross Anderson}, \bibinfo{person}{Itai
  Ashlagi}, \bibinfo{person}{David Gamarnik}, {and} \bibinfo{person}{Alvin~E
  Roth}.} \bibinfo{year}{2015}\natexlab{}.
\newblock \showarticletitle{Finding Long Chains in Kidney Exchange Using the
  Travelling Salesman Problem}. In \bibinfo{booktitle}{\emph{Proceedings of the
  National Academy of Sciences}}, Vol.~\bibinfo{volume}{112}.
  \bibinfo{publisher}{National Academy of Sciences}.
\newblock
\newblock
\shownote{\url{https://doi.org/10.1073/pnas.1421853112}}.


\bibitem[Araki et~al\mbox{.}(2016)]%
        {Araki_ReplicatedSecretSharing_2016}
\bibfield{author}{\bibinfo{person}{Toshinori Araki}, \bibinfo{person}{Jun
  Furukawa}, \bibinfo{person}{Yehuda Lindell}, \bibinfo{person}{Ariel Nof},
  {and} \bibinfo{person}{Kazuma Ohara}.} \bibinfo{year}{2016}\natexlab{}.
\newblock \showarticletitle{High-Throughput Semi-Honest Secure Three-Party
  Computation with an Honest Majority}. In \bibinfo{booktitle}{\emph{Computer
  and Communications Security}}. \bibinfo{publisher}{ACM}.
\newblock
\newblock
\shownote{\url{https://doi.org/10.1145/2976749.2978331}}.


\bibitem[Araki et~al\mbox{.}(2021)]%
        {Araki_SecureGraphAnalysis_2021}
\bibfield{author}{\bibinfo{person}{Toshinori Araki}, \bibinfo{person}{Jun
  Furukawa}, \bibinfo{person}{Kazuma Ohara}, \bibinfo{person}{Benny Pinkas},
  \bibinfo{person}{Hanan Rosemarin}, {and} \bibinfo{person}{Hikaru Tsuchida}.}
  \bibinfo{year}{2021}\natexlab{}.
\newblock \showarticletitle{Secure Graph Analysis at Scale}. In
  \bibinfo{booktitle}{\emph{Computer and Communications Security}}.
  \bibinfo{publisher}{ACM}.
\newblock
\newblock
\shownote{\url{https://doi.org/10.1145/3460120.3484560}}.


\bibitem[Ashlagi et~al\mbox{.}(2018)]%
        {AshlagiFrequencies2018}
\bibfield{author}{\bibinfo{person}{Itai Ashlagi}, \bibinfo{person}{Adam
  Bingaman}, \bibinfo{person}{Maximilien Burq}, \bibinfo{person}{Vahideh
  Manshadi}, \bibinfo{person}{David Gamarnik}, \bibinfo{person}{Cathi Murphey},
  \bibinfo{person}{Alvin~E Roth}, \bibinfo{person}{Marc~L Melcher}, {and}
  \bibinfo{person}{Michael~A Rees}.} \bibinfo{year}{2018}\natexlab{}.
\newblock \showarticletitle{Effect of Match-Run Frequencies on the Number of
  Transplants and Waiting Times in Kidney Exchange}. In
  \bibinfo{booktitle}{\emph{American Journal of Transplantation}},
  Vol.~\bibinfo{volume}{18}. \bibinfo{publisher}{Wiley Online Library}.
\newblock
\newblock
\shownote{\url{https://doi.org/10.1111/ajt.14566}}.


\bibitem[Ashlagi and Roth(2021)]%
        {AshlagiKidneyExchangeOperations2021}
\bibfield{author}{\bibinfo{person}{Itai Ashlagi} {and} \bibinfo{person}{Alvin~E
  Roth}.} \bibinfo{year}{2021}\natexlab{}.
\newblock \showarticletitle{Kidney Exchange: An Operations Perspective}. In
  \bibinfo{booktitle}{\emph{Management Science}}, Vol.~\bibinfo{volume}{67}.
  \bibinfo{publisher}{INFORMS}.
\newblock
\newblock
\shownote{\url{https://doi.org/10.1287/mnsc.2020.3954}}.


\bibitem[Bastos et~al\mbox{.}(2021)]%
        {Bastos_KEBrazil_2021}
\bibfield{author}{\bibinfo{person}{Juliana Bastos}, \bibinfo{person}{Michal
  Mankowski}, \bibinfo{person}{A Massie}, \bibinfo{person}{M Levan},
  \bibinfo{person}{C Bisi}, \bibinfo{person}{C Stopato}, \bibinfo{person}{T
  Freesz}, \bibinfo{person}{V Colares}, \bibinfo{person}{G Ferreira},
  {et~al\mbox{.}}} \bibinfo{year}{2021}\natexlab{}.
\newblock \showarticletitle{Kidney Paired Donation in Brazil-A Single Center
  Perspective}. In \bibinfo{booktitle}{\emph{Transplant International}},
  Vol.~\bibinfo{volume}{34}. \bibinfo{publisher}{Wiley Online Library}.
\newblock
\newblock
\shownote{\url{https://doi.org/10.1111/tri.13923}}.


\bibitem[Beaver(1992)]%
        {Beaver_Triples_1992}
\bibfield{author}{\bibinfo{person}{Donald Beaver}.}
  \bibinfo{year}{1992}\natexlab{}.
\newblock \showarticletitle{Efficient Multiparty Protocols using Circuit
  Randomization}. In \bibinfo{booktitle}{\emph{Advances in
  Cryptology—CRYPTO}}. \bibinfo{publisher}{Springer}.
\newblock
\newblock
\shownote{\url{https://doi.org/10.1007/3-540-46766-1_34}}.


\bibitem[Birka et~al\mbox{.}(2022)]%
        {Birka_PPKidneyExchange_2022}
\bibfield{author}{\bibinfo{person}{Timm Birka}, \bibinfo{person}{Kay Hamacher},
  \bibinfo{person}{Tobias Kussel}, \bibinfo{person}{Helen M{\"o}llering}, {and}
  \bibinfo{person}{Thomas Schneider}.} \bibinfo{year}{2022}\natexlab{}.
\newblock \showarticletitle{SPIKE: Secure and Private Investigation of the
  Kidney Exchange Problem}. In \bibinfo{booktitle}{\emph{BMC Medical
  Informatics and Decision Making}}, Vol.~\bibinfo{volume}{22}.
  \bibinfo{publisher}{Springer}.
\newblock
\newblock
\shownote{\url{https://doi.org/10.1186/s12911-022-01994-4}}.


\bibitem[Bir{\'o} et~al\mbox{.}(2019)]%
        {Biro_EuropeanModellingKE_2019}
\bibfield{author}{\bibinfo{person}{P{\'e}ter Bir{\'o}}, \bibinfo{person}{Joris
  van~de Klundert}, \bibinfo{person}{David Manlove}, \bibinfo{person}{William
  Pettersson}, \bibinfo{person}{Tommy Andersson}, \bibinfo{person}{Lisa
  Burnapp}, \bibinfo{person}{Pavel Chromy}, \bibinfo{person}{Pablo Delgado},
  \bibinfo{person}{Piotr Dworczak}, \bibinfo{person}{Bernadette Haase},
  {et~al\mbox{.}}} \bibinfo{year}{2019}\natexlab{}.
\newblock \showarticletitle{Modelling and Optimisation in European Kidney
  Exchange Programmes}. In \bibinfo{booktitle}{\emph{European Journal of
  Operational Research}}. \bibinfo{publisher}{Elsevier}.
\newblock
\newblock
\shownote{\url{https://doi.org/10.1016/j.ejor.2019.09.006}}.


\bibitem[Bofill et~al\mbox{.}(2017)]%
        {Bofill_SpanishKEP_2017}
\bibfield{author}{\bibinfo{person}{Miquel Bofill}, \bibinfo{person}{Marcos
  Calder{\'o}n}, \bibinfo{person}{Francesc Castro}, \bibinfo{person}{Esteve
  Del~Acebo}, \bibinfo{person}{Pablo Delgado}, \bibinfo{person}{Marc Garcia},
  \bibinfo{person}{Marta Garc{\'\i}a}, \bibinfo{person}{Marc Roig},
  \bibinfo{person}{Mar{\'\i}a~O Valent{\'\i}n}, {and} \bibinfo{person}{Mateu
  Villaret}.} \bibinfo{year}{2017}\natexlab{}.
\newblock \showarticletitle{The Spanish Kidney Exchange Model: Study of
  Computation-based Alternatives to the Current Procedure}. In
  \bibinfo{booktitle}{\emph{Artificial Intelligence in Medicine}}.
  \bibinfo{publisher}{Springer}.
\newblock
\newblock
\shownote{\url{https://doi.org/10.1007/978-3-319-59758-4_31}}.


\bibitem[Boyle et~al\mbox{.}(2021)]%
        {Boyle_FssMixed_2021}
\bibfield{author}{\bibinfo{person}{Elette Boyle}, \bibinfo{person}{Nishanth
  Chandran}, \bibinfo{person}{Niv Gilboa}, \bibinfo{person}{Divya Gupta},
  \bibinfo{person}{Yuval Ishai}, \bibinfo{person}{Nishant Kumar}, {and}
  \bibinfo{person}{Mayank Rathee}.} \bibinfo{year}{2021}\natexlab{}.
\newblock \showarticletitle{Function Secret Sharing for Mixed-Mode and
  Fixed-Point Secure Computation}. In \bibinfo{booktitle}{\emph{Advances in
  Cryptology – EUROCRYPT}}. \bibinfo{publisher}{Springer}.
\newblock
\newblock
\shownote{\url{https://doi.org/10.1007/978-3-030-77886-6_30}}.


\bibitem[Boyle et~al\mbox{.}(2015)]%
        {Boyle_Fss_2015}
\bibfield{author}{\bibinfo{person}{Elette Boyle}, \bibinfo{person}{Niv Gilboa},
  {and} \bibinfo{person}{Yuval Ishai}.} \bibinfo{year}{2015}\natexlab{}.
\newblock \showarticletitle{Function Secret Sharing}. In
  \bibinfo{booktitle}{\emph{Advances in Cryptology - EUROCRYPT}}.
  \bibinfo{publisher}{Springer}.
\newblock
\newblock
\shownote{\url{https://doi.org/10.1007/978-3-662-46803-6_12}}.


\bibitem[Breuer et~al\mbox{.}(2022a)]%
        {Breuer_KepIp_2022}
\bibfield{author}{\bibinfo{person}{Malte Breuer}, \bibinfo{person}{Pascal
  Hein}, \bibinfo{person}{Leonardo Pompe}, \bibinfo{person}{Ben Temme},
  \bibinfo{person}{Ulrike Meyer}, {and} \bibinfo{person}{Susanne Wetzel}.}
  \bibinfo{year}{2022}\natexlab{a}.
\newblock \showarticletitle{Solving the Kidney Exchange Problem Using
  Privacy-Preserving Integer Programming}. In \bibinfo{booktitle}{\emph{19th
  Annual International Conference on Privacy, Security \& Trust}}.
  \bibinfo{publisher}{IEEE}.
\newblock
\newblock
\shownote{\url{https://doi.org/10.1109/PST55820.2022.9851968}}.


\bibitem[Breuer et~al\mbox{.}(2023)]%
        {Breuer_KepIpExtended_2022}
\bibfield{author}{\bibinfo{person}{Malte Breuer}, \bibinfo{person}{Pascal
  Hein}, \bibinfo{person}{Leonardo Pompe}, \bibinfo{person}{Ben Temme},
  \bibinfo{person}{Ulrike Meyer}, {and} \bibinfo{person}{Susanne Wetzel}.}
  \bibinfo{year}{2023}\natexlab{}.
\newblock \showarticletitle{Solving the Kidney Exchange Problem Using
  Privacy-Preserving Integer Programming (Updated and Extended Version)}. In
  \bibinfo{booktitle}{\emph{arXiv preprint arXiv:2208.11319}}.
\newblock
\newblock
\shownote{\url{https://arxiv.org/abs/2208.11319}}.


\bibitem[Breuer et~al\mbox{.}(2022b)]%
        {Breuer_Matching_2022}
\bibfield{author}{\bibinfo{person}{Malte Breuer}, \bibinfo{person}{Ulrike
  Meyer}, {and} \bibinfo{person}{Susanne Wetzel}.}
  \bibinfo{year}{2022}\natexlab{b}.
\newblock \showarticletitle{Privacy-Preserving Maximum Matching on General
  Graphs and its Application to Enable Privacy-Preserving Kidney Exchange}. In
  \bibinfo{booktitle}{\emph{Twelveth ACM Conference on Data and Application
  Security and Privacy}}. \bibinfo{publisher}{ACM}.
\newblock
\newblock
\shownote{\url{https://doi.org/10.1145/3508398.3511509}}.


\bibitem[Breuer et~al\mbox{.}(2020)]%
        {Breuer_KEprotocol_2020}
\bibfield{author}{\bibinfo{person}{Malte Breuer}, \bibinfo{person}{Ulrike
  Meyer}, \bibinfo{person}{Susanne Wetzel}, {and} \bibinfo{person}{Anja
  M{\"u}hlfeld}.} \bibinfo{year}{2020}\natexlab{}.
\newblock \showarticletitle{A Privacy-Preserving Protocol for the Kidney
  Exchange Problem}. In \bibinfo{booktitle}{\emph{Workshop on Privacy in the
  Electronic Society}}. \bibinfo{publisher}{ACM}.
\newblock
\newblock
\shownote{\url{https://doi.org/10.1145/3411497.3420213}}.


\bibitem[Br{\"u}ggemann et~al\mbox{.}(2022)]%
        {Bruggemann_MaxMatchingApproximation_2022}
\bibfield{author}{\bibinfo{person}{Andreas Br{\"u}ggemann},
  \bibinfo{person}{Malte Breuer}, \bibinfo{person}{Andreas Klinger},
  \bibinfo{person}{Thomas Schneider}, {and} \bibinfo{person}{Ulrike Meyer}.}
  \bibinfo{year}{2022}\natexlab{}.
\newblock \showarticletitle{Secure Maximum Weight Matching Approximation on
  General Graphs}. In \bibinfo{booktitle}{\emph{Workshop on Privacy in the
  Electronic Society}}.
\newblock
\newblock
\shownote{\url{https://doi.org/10.1145/3559613.3563209}}.


\bibitem[Canetti(2001)]%
        {Canetti_UC_2001}
\bibfield{author}{\bibinfo{person}{Ran Canetti}.}
  \bibinfo{year}{2001}\natexlab{}.
\newblock \showarticletitle{Universally Composable Security: A New Paradigm for
  Cryptographic Protocols}. In \bibinfo{booktitle}{\emph{Symposium on
  Foundations of Computer Science}}. \bibinfo{publisher}{IEEE}.
\newblock
\newblock
\shownote{\url{https://doi.org/10.1109/SFCS.2001.959888}}.


\bibitem[Catrina and De~Hoogh(2010)]%
        {Catrina_PrimitivesSMPC_2010}
\bibfield{author}{\bibinfo{person}{Octavian Catrina} {and}
  \bibinfo{person}{Sebastiaan De~Hoogh}.} \bibinfo{year}{2010}\natexlab{}.
\newblock \showarticletitle{Improved Primitives for Secure Multiparty Integer
  Computation}. In \bibinfo{booktitle}{\emph{International Conference on
  Security and Cryptography for Networks}}. \bibinfo{publisher}{Springer}.
\newblock
\newblock
\shownote{\url{https://doi.org/10.1007/978-3-642-15317-4_13}}.


\bibitem[Constantino et~al\mbox{.}(2013)]%
        {Constantino_KEP_2013}
\bibfield{author}{\bibinfo{person}{Miguel Constantino}, \bibinfo{person}{Xenia
  Klimentova}, \bibinfo{person}{Ana Viana}, {and} \bibinfo{person}{Abdur
  Rais}.} \bibinfo{year}{2013}\natexlab{}.
\newblock \showarticletitle{New Insights on Integer-Programming Models for the
  Kidney Exchange Problem}. In \bibinfo{booktitle}{\emph{European Journal of
  Operational Research}}, Vol.~\bibinfo{volume}{231}.
  \bibinfo{publisher}{Elsevier}.
\newblock
\newblock
\shownote{\url{https://doi.org/10.1016/j.ejor.2013.05.025}}.


\bibitem[Cramer et~al\mbox{.}(2018)]%
        {Cramer_Spdz2k_2018}
\bibfield{author}{\bibinfo{person}{Ronald Cramer}, \bibinfo{person}{Ivan
  Damg{\aa}rd}, \bibinfo{person}{Daniel Escudero}, \bibinfo{person}{Peter
  Scholl}, {and} \bibinfo{person}{Chaoping Xing}.}
  \bibinfo{year}{2018}\natexlab{}.
\newblock \showarticletitle{SPD{$ \ints_{2^k} $}: Efficient MPC mod {$ 2^k $}
  for Dishonest Majority}. In \bibinfo{booktitle}{\emph{Advances in Cryptology
  -- CRYPTO}}. \bibinfo{publisher}{Springer}.
\newblock
\newblock
\shownote{\url{https://doi.org/10.1007/978-3-319-96881-0_26}}.


\bibitem[Damg{\aa}rd et~al\mbox{.}(2019)]%
        {Damgard_NewPrimitivesActiveSecurity_2019}
\bibfield{author}{\bibinfo{person}{Ivan Damg{\aa}rd}, \bibinfo{person}{Daniel
  Escudero}, \bibinfo{person}{Tore Frederiksen}, \bibinfo{person}{Marcel
  Keller}, \bibinfo{person}{Peter Scholl}, {and} \bibinfo{person}{Nikolaj
  Volgushev}.} \bibinfo{year}{2019}\natexlab{}.
\newblock \showarticletitle{New Primitives for Actively-Secure MPC over Rings
  with Applications to Private Machine Learning}. In
  \bibinfo{booktitle}{\emph{Symposium on Security and Privacy}}.
  \bibinfo{publisher}{IEEE}.
\newblock
\newblock
\shownote{\url{https://doi.org/10.1109/SP.2019.00078}}.


\bibitem[Demmler et~al\mbox{.}(2015)]%
        {Demmler_ABY_2015}
\bibfield{author}{\bibinfo{person}{Daniel Demmler}, \bibinfo{person}{Thomas
  Schneider}, {and} \bibinfo{person}{Michael Zohner}.}
  \bibinfo{year}{2015}\natexlab{}.
\newblock \showarticletitle{ABY-A Framework for Efficient Mixed-Protocol Secure
  Two-Party Computation.}. In \bibinfo{booktitle}{\emph{Network and Distributed
  System Security Symposium}}. \bibinfo{publisher}{The Internet Society}.
\newblock
\newblock
\shownote{\url{https://doi.org/10.14722/ndss.2015.23113}}.


\bibitem[{DESMO-J}(enet)]%
        {DESMOJ}
\bibfield{author}{\bibinfo{person}{{DESMO-J}}.}
  \bibinfo{year}{\url{http://desmoj.sourceforge.net}}\natexlab{}.
\newblock \bibinfo{title}{Accessed 03-May-2021}.
\newblock
\newblock


\bibitem[Dickerson et~al\mbox{.}(2012)]%
        {Dickerson_Chains_2012}
\bibfield{author}{\bibinfo{person}{John~P. Dickerson}, \bibinfo{person}{Ariel~D
  Procaccia}, {and} \bibinfo{person}{Tuomas Sandholm}.}
  \bibinfo{year}{2012}\natexlab{}.
\newblock \showarticletitle{Optimizing Kidney Exchange with Transplant Chains:
  Theory and Reality}. In \bibinfo{booktitle}{\emph{International Conference on
  Autonomous Agents and Multiagent Systems-Volume 2}}.
  \bibinfo{publisher}{ACM}.
\newblock


\bibitem[Dickerson et~al\mbox{.}(2019)]%
        {Dickerson_FailureAwareKE_2019}
\bibfield{author}{\bibinfo{person}{John~P. Dickerson},
  \bibinfo{person}{Ariel~D. Procaccia}, {and} \bibinfo{person}{Tuomas
  Sandholm}.} \bibinfo{year}{2019}\natexlab{}.
\newblock \showarticletitle{Failure-Aware Kidney Exchange}. In
  \bibinfo{booktitle}{\emph{Management Science}}, Vol.~\bibinfo{volume}{65}.
  \bibinfo{publisher}{INFORMS}.
\newblock
\newblock
\shownote{\url{https://doi.org/10.1287/mnsc.2018.3026}}.


\bibitem[Eerikson et~al\mbox{.}(2020)]%
        {Eerikson_UseYourBrain_2020}
\bibfield{author}{\bibinfo{person}{Hendrik Eerikson}, \bibinfo{person}{Marcel
  Keller}, \bibinfo{person}{Claudio Orlandi}, \bibinfo{person}{Pille Pullonen},
  \bibinfo{person}{Joonas Puura}, {and} \bibinfo{person}{Mark Simkin}.}
  \bibinfo{year}{2020}\natexlab{}.
\newblock \showarticletitle{{Use Your Brain! Arithmetic 3PC for Any Modulus
  with Active Security}}. In \bibinfo{booktitle}{\emph{1st Conference on
  Information-Theoretic Cryptography (ITC 2020)}}, Vol.~\bibinfo{volume}{163}.
  \bibinfo{publisher}{Schloss Dagstuhl -- Leibniz-Zentrum f{\"u}r Informatik}.
\newblock
\newblock
\shownote{\url{https://doi.org/10.4230/LIPIcs.ITC.2020.5}}.


\bibitem[Gupta et~al\mbox{.}(2022)]%
        {Gupta_Llama_2022}
\bibfield{author}{\bibinfo{person}{Kanav Gupta}, \bibinfo{person}{Deepak
  Kumaraswamy}, \bibinfo{person}{Nishanth Chandran}, {and}
  \bibinfo{person}{Divya Gupta}.} \bibinfo{year}{2022}\natexlab{}.
\newblock \showarticletitle{LLAMA: A Low Latency Math Library for Secure
  Inference}. In \bibinfo{booktitle}{\emph{Proceedings on Privacy Enhancing
  Technologies}}, Vol.~\bibinfo{volume}{4}.
\newblock
\newblock
\shownote{\url{https://doi.org/10.56553/popets-2022-0109}}.


\bibitem[Jawalkar et~al\mbox{.}(2023)]%
        {Jawalker_FssPPML_2023}
\bibfield{author}{\bibinfo{person}{Neha Jawalkar}, \bibinfo{person}{Kanav
  Gupta}, \bibinfo{person}{Arkaprava Basu}, \bibinfo{person}{Nishanth
  Chandran}, \bibinfo{person}{Divya Gupta}, {and} \bibinfo{person}{Rahul
  Sharma}.} \bibinfo{year}{2023}\natexlab{}.
\newblock \showarticletitle{Orca: FSS-based Secure Training with GPUs}. In
  \bibinfo{booktitle}{\emph{Cryptology ePrint Archive}}.
\newblock


\bibitem[Keller(2020)]%
        {KellerMPSPDZ2020}
\bibfield{author}{\bibinfo{person}{Marcel Keller}.}
  \bibinfo{year}{2020}\natexlab{}.
\newblock \showarticletitle{{MP-SPDZ}: A Versatile Framework for Multi-Party
  Computation}. In \bibinfo{booktitle}{\emph{Computer and Communications
  Security}}. \bibinfo{publisher}{ACM}.
\newblock
\newblock
\shownote{\url{https://doi.org/10.1145/3372297.3417872}}.


\bibitem[Keller and Scholl(2014)]%
        {Keller_EfficientDataStructures_2014}
\bibfield{author}{\bibinfo{person}{Marcel Keller} {and} \bibinfo{person}{Peter
  Scholl}.} \bibinfo{year}{2014}\natexlab{}.
\newblock \showarticletitle{Efficient, Oblivious Data Structures for MPC}. In
  \bibinfo{booktitle}{\emph{Advances in Cryptology--ASIACRYPT}}.
  \bibinfo{publisher}{Springer}.
\newblock
\newblock
\shownote{\url{https://doi.org/10.1007/978-3-662-45608-8_27}}.


\bibitem[Launchbury et~al\mbox{.}(2012)]%
        {Launchbury_Multiplex_2012}
\bibfield{author}{\bibinfo{person}{John Launchbury}, \bibinfo{person}{Iavor~S
  Diatchki}, \bibinfo{person}{Thomas DuBuisson}, {and} \bibinfo{person}{Andy
  Adams-Moran}.} \bibinfo{year}{2012}\natexlab{}.
\newblock \showarticletitle{Efficient Lookup-Table Protocol in Secure
  Multiparty Computation}. In \bibinfo{booktitle}{\emph{International
  Conference on Functional Programming}}. \bibinfo{publisher}{ACM}.
\newblock
\newblock
\shownote{\url{https://doi.org/10.1145/2364527.2364556}}.


\bibitem[Lindell and Nof(2017)]%
        {Lindell_ArakiFields_2017}
\bibfield{author}{\bibinfo{person}{Yehuda Lindell} {and} \bibinfo{person}{Ariel
  Nof}.} \bibinfo{year}{2017}\natexlab{}.
\newblock \showarticletitle{A Framework for Constructing Fast MPC over
  Arithmetic Circuits with Malicious Adversaries and an Honest-Majority}. In
  \bibinfo{booktitle}{\emph{Computer and Communications Security}}.
  \bibinfo{publisher}{ACM}.
\newblock
\newblock
\shownote{\url{https://doi.org/10.1145/3133956.3133999}}.


\bibitem[Pape and Conradt(1980)]%
        {PapeMatching1980}
\bibfield{author}{\bibinfo{person}{U. Pape} {and} \bibinfo{person}{D.
  Conradt}.} \bibinfo{year}{1980}\natexlab{}.
\newblock \showarticletitle{Maximales Matching in Graphen}. In
  \bibinfo{booktitle}{\emph{Ausgew{\"a}hlte Operations Research Software in
  FORTRAN}}.
\newblock


\bibitem[Roth et~al\mbox{.}(2022)]%
        {Roth_IllegalCountries_2022}
\bibfield{author}{\bibinfo{person}{Alvin~E Roth}, \bibinfo{person}{Ignazio~R
  Marino}, \bibinfo{person}{Kimberly~D Krawiec}, {and}
  \bibinfo{person}{Michael~A Rees}.} \bibinfo{year}{2022}\natexlab{}.
\newblock \showarticletitle{Criminal, Legal, and Ethical Kidney Donation and
  Transplantation: A Conceptual Framework to Enable Innovation}. In
  \bibinfo{booktitle}{\emph{Transplant International}},
  Vol.~\bibinfo{volume}{35}. \bibinfo{publisher}{Frontiers}.
\newblock
\newblock
\shownote{\url{https://doi.org/10.3389/ti.2022.10551}}.


\bibitem[Roth et~al\mbox{.}(2007)]%
        {RothEfficientKidneyExchange2007}
\bibfield{author}{\bibinfo{person}{Alvin~E. Roth}, \bibinfo{person}{Tayfun
  S{\"o}nmez}, {and} \bibinfo{person}{M.~Utku {\"U}nver}.}
  \bibinfo{year}{2007}\natexlab{}.
\newblock \showarticletitle{Efficient Kidney Exchange: Coincidence of Wants in
  Markets with Compatibility-Based Preferences}. In
  \bibinfo{booktitle}{\emph{American Economic Review}},
  Vol.~\bibinfo{volume}{97}.
\newblock
\newblock
\shownote{\url{https://doi.org/10.1257/aer.97.3.828}}.


\bibitem[{Scandiatransplant}(torg)]%
        {Scandiatransplant_2023}
\bibfield{author}{\bibinfo{person}{{Scandiatransplant}}.}
  \bibinfo{year}{\url{http://www.scandiatransplant.org/}}\natexlab{}.
\newblock \bibinfo{title}{Accessed 17-Aug-2023}.
\newblock
\newblock


\bibitem[Storrier et~al\mbox{.}(2023)]%
        {Storrier_Grotto_2023}
\bibfield{author}{\bibinfo{person}{Kyle Storrier}, \bibinfo{person}{Adithya
  Vadapalli}, \bibinfo{person}{Allan Lyons}, {and} \bibinfo{person}{Ryan
  Henry}.} \bibinfo{year}{2023}\natexlab{}.
\newblock \showarticletitle{Grotto: Screaming Fast ({$ 2+1 $})-PC or {$
  \ints_{2^n} $} via ($ 2, 2 $)-DPFs}. In \bibinfo{booktitle}{\emph{Computer
  and Communications Security}}. \bibinfo{publisher}{ACM}.
\newblock
\newblock
\shownote{\url{https://doi.org/10.1145/3576915.3623147}}.


\bibitem[{Transplantationsgesetz §8 (German Transplantation Law)}(Gpdf)]%
        {TPG}
\bibfield{author}{\bibinfo{person}{{Transplantationsgesetz §8 (German
  Transplantation Law)}}.}
  \bibinfo{year}{\url{https://www.gesetze-im-internet.de/tpg/TPG.pdf}}\natexlab{}.
\newblock \bibinfo{title}{Accessed 17-Aug-2023}.
\newblock
\newblock


\bibitem[{United Network for Organ Sharing. The Kidney Transplant Waiting
  List}(list)]%
        {UNOS_WaitingList_2023}
\bibfield{author}{\bibinfo{person}{{United Network for Organ Sharing. The
  Kidney Transplant Waiting List}}.}
  \bibinfo{year}{\url{https://transplantliving.org/kidney/the-kidney-transplant-waitlist/}}\natexlab{}.
\newblock \bibinfo{title}{Acc. 17-Aug-2023}.
\newblock
\newblock


\bibitem[Wagh(2022)]%
        {Wagh_FssPPML_2022}
\bibfield{author}{\bibinfo{person}{Sameer Wagh}.}
  \bibinfo{year}{2022}\natexlab{}.
\newblock \showarticletitle{Pika: Secure Computation using Function Secret
  Sharing over Rings}. In \bibinfo{booktitle}{\emph{Proceedings on Privacy
  Enhancing Technologies}}, Vol.~\bibinfo{volume}{4}.
\newblock
\newblock
\shownote{\url{https://doi.org/10.56553/popets-2022-0113}}.


\end{thebibliography}

\appendix

\section{Proof of the Approximation Ratio}\label{app:proof}

\begin{theorem}\label{thm:approx_ratio}
	Algorithm~\ref{alg:greedy} computes a 1/3-approximation of the KEP with a maximum cycle size of $ 3 $ (cf.~Definition~\ref{def:kep}).
\end{theorem}
\begin{proof}
	Let $ \proofOptCycles $ be a set of cycles for an optimal solution for the KEP and let $ \proofCompCycles $ be a set of cycles computed by Algorithm~\ref{alg:greedy}. 
	
	For all $ \proofCompCyclesSingle \in \proofCompCycles $, let $ \proofOptCycles_{\proofCompCyclesSingle} $ be the subset of $ \proofOptCycles $ containing all cycles in the optimal solution that share at least one node with $ \proofCompCyclesSingle $, i.e., $ \proofOptCycles_{\proofCompCyclesSingle} = \{\proofOptCyclesSingle \in  \proofOptCycles \ \vert \ \proofOptCyclesSingle \cap \proofCompCyclesSingle \neq \emptyset\} $. 
	
	Then, $ \forall \proofCompCyclesSingle \in \proofCompCycles $ it holds that $ \vert \proofOptCycles_{\proofCompCyclesSingle} \vert \leq 3 $ since in the worst case $ \proofCompCyclesSingle = (v_1, v_2, v_3) $ such that all nodes of $ \proofCompCyclesSingle $ are part of a different cycle in the set of cycles $ \proofOptCycles $ for the optimal solution.
	
	Let $ \proofOptCyclesSingle_{v_1}, \proofOptCyclesSingle_{v_2}, \proofOptCyclesSingle_{v_3} $ be the three cycles that were chosen instead of $ \proofCompCyclesSingle$ in the optimal solution. 
	We know that for the weights of these cycles it holds that $ \weightFunction(\proofOptCyclesSingle_{v_i}) \leq \weightFunction(\proofCompCyclesSingle) $ with $ i \in \{1, 2, 3\} $ since Algorithm~\ref{alg:greedy} would not have chosen $ \proofCompCyclesSingle $ if there was an alternative cycle of larger weight.
	Thus, it holds that $ \forall \proofCompCyclesSingle \in \proofCompCycles: \ \sum_{\proofOptCyclesSingle \in \proofOptCycles_{\proofCompCyclesSingle}} \weightFunction(\proofOptCyclesSingle) \leq 3 \ \weightFunction(\proofCompCyclesSingle) $.
	
	Overall, we obtain: 
	\begin{equation*}
		\sum_{\proofOptCyclesSingle \in \proofOptCycles} \weightFunction(\proofOptCyclesSingle) \leq \sum_{\proofCompCyclesSingle \in \proofCompCycles}\sum_{\proofOptCyclesSingle \in \proofOptCycles_{\proofCompCyclesSingle}} \weightFunction(\proofOptCyclesSingle) \leq 3 \sum_{\proofCompCyclesSingle \in \proofCompCycles} \weightFunction(\proofCompCyclesSingle)
	\end{equation*}
\end{proof}

\section{Demux Gate}\label{app:gates}

In this section, we provide the specification of the gate~$ \demuxGate{n} $ that we use in our protocol to transform a value $ x $ into a binary indicator vector $ \enc{d} $ of size $ n $, where $ n $ is the maximum possible value for $ x $. For the output $ \enc{d} $, it holds that $ \enc{d}(i) = \enc{1} $ if $ i = x $, and $ \enc{d}(i) = \enc{0} $, otherwise. The gate is based on the approach from~\cite{Launchbury_Multiplex_2012} and the implementation of $ \demuxVecGate $ provided in Gate~\ref{gat:demux} is taken from~\cite{Keller_EfficientDataStructures_2014}.

\begin{gate}[t]
	\algrenewcommand\algorithmicindent{0.7em}
\small
\begin{flushleft}
	\textbf{Input:} Secret value $ \enc{x} $ and $ n $, which is the maximum possible value of $ x $ \vspace*{0.2em} \\
	\textbf{Output:} Binary indicator vector $ \enc{d} $ of length $ n $
\end{flushleft}
\begin{algorithmic}[1]
	\State Generate the bit decomposition $ \enc{x}(0), ..., \enc{x}(\bitsize - 1) $ of $ \enc{x} $
	\State $ \enc{d} \leftarrow \demuxVecGate(\enc{x}(0), ..., \enc{x}(\bitsize - 1)) $
	\State \Return $ (\enc{d}(0), ..., \enc{d}(n - 1)) $
	\State 
	\Procedure{demux-vec}{x(0), ..., x(l)}
		\If{$ l = 1 $}
			\State \Return $ (1 - \enc{x}(0), \enc{x}(0)) $
		\EndIf
		\State $ \enc{y}(0), ..., \enc{y}(2^{\lfloor l / 2 \rfloor} - 1) \leftarrow \demuxVecGate{\enc{x}(0),...,\enc{x}(\lfloor l / 2 \rfloor - 1)} $
		\State $ \enc{z}(0), ..., \enc{z}(2^{\lfloor l / 2 \rfloor} - 1) \leftarrow \demuxVecGate{\enc{x}(\lfloor l / 2 \rfloor),...,\enc{x}(b - 1)} $
		\For{$ 0 \leq i < 2^l - 1 $}
			\State $ \enc{d}[i] \leftarrow \enc{y}(i \mod 2^{\lfloor l / 2 \rfloor}) \cdot \enc{z}(i \mod 2^{\lfloor l / 2 \rfloor}) $
		\EndFor
		\State \Return $ (\enc{d}(0), ..., \enc{d}(2^l - 1)) $
	\EndProcedure
\end{algorithmic}
\vspace*{-0.1em}

	\caption{$ \demuxGate{n} $, based on~\cite{Keller_EfficientDataStructures_2014}}
	\label{gat:demux}
\end{gate}

As a first step, $ \enc{x} $ is transformed into its binary representation. Then, the protocol from~\cite{Keller_EfficientDataStructures_2014} is used to create the binary indicator vector $ \enc{d} $ from the binary representation of $ \enc{x} $. The gate~$ \demuxGate{n} $ then returns the first $ n $ entries of $ \enc{d} $.

In~\cite{Keller_EfficientDataStructures_2014}, the authors show that the conversion of the binary representation of $ x $ into the binary indicator vector $ d $ requires $ \bigo{2^\bitsize} $ invocations in $ \bigo{\log \bitsize} $ rounds, where $ \bitsize $ is the bitsize of a secret value. These are also the overall complexities of the gate~$ \demuxGate{n} $.

\section{SMPC Primitives}\label{app:smpc_protocols}

In this section, we provide further details on the different SMPC primitives that we use in the run time evaluation of our protocol~$ \kepApProtocol $ (Section~\ref{sub:eval_runtime}). 

The SMPC primitive~\texttt{rep-ring} uses the three-party replicated secret sharing scheme by Araki et al.~\cite{Araki_ReplicatedSecretSharing_2016}. To the best of our knowledge, this is the most efficient scheme to date for three parties in the semi-honest model with an honest majority. It can be instantiated both over the field~$ \field_p $ as well as the ring $ \ints_{2^k} $. We use the ring setting since it yields more efficient run times for our protocol~$ \kepApProtocol $.

Lindell and Nof~\cite{Lindell_ArakiFields_2017} have extended the scheme from~\cite{Araki_ReplicatedSecretSharing_2016} to the malicious model for the field setting. 
The basic idea is to generate potentially incorrect multiplication triples~\cite{Beaver_Triples_1992} during a pre-processing phase and then use these in a post-sacrifice step to check that all multiplications were executed correctly during the protocol execution. 
The SMPC primitive~\texttt{ps-rep-ring} is a further extension of the approach from~\cite{Lindell_ArakiFields_2017} by Eerikson et al.~\cite{Eerikson_UseYourBrain_2020} to the ring setting.

The third SMPC primitive that we use is \texttt{semi2k} which is a semi-honest version of the SMPC primitive \texttt{spdz2k} by Cramer et al.~\cite{Cramer_Spdz2k_2018}, which uses oblivious transfer to generate shares of authenticated multiplication triples. In contrast to the SMPC primitives~\texttt{rep-ring} and \texttt{ps-rep-ring}, \texttt{semi2k} provides for security in the presence of a dishonest majority.

\section{Evaluation for Crossover Exchange}\label{app:twocycles}

While most existing kidney exchange platforms consider a maximum cycle size of $ 3 $, i.e., both cycles of size $ 2 $ and $ 3 $, there are also some countries (e.g., all countries of Scandiatransplant~\footnote{Scandiatransplant comprises the countries Denmark, Finland, Iceland, Norway, Sweden and Estonia~\cite{Scandiatransplant_2023}.}) that only allow crossover exchanges, i.e., exchange cycles of size $ 2 $.
In this section, we present the results for our protocol~$ \kepApProtocol $ when considering crossover exchanges only. For this special case, some parts of the protocol~$ \kepApProtocol $ (cf.~Protocol~\ref{prot:kep_approx}) can be simplified as the set $ \subsets $ of subsets now only contains all subsets of size $ 2 $ and thus it is no longer necessary to represent a subset by three nodes. In particular, this means that the distinction between subsets of size $ 2 $ and $ 3 $ in Steps~$ 13-17 $ and in Steps~$ 35-41 $ of Protocol~\ref{prot:kep_approx} can be skipped and that the matrix $ \enc{\subsetCycleMap} $ is no longer required since there is only a single possible exchange cycle for each subset of size $ 2 $.

\begin{figure}[t]
	\includegraphics*[scale=1.0]{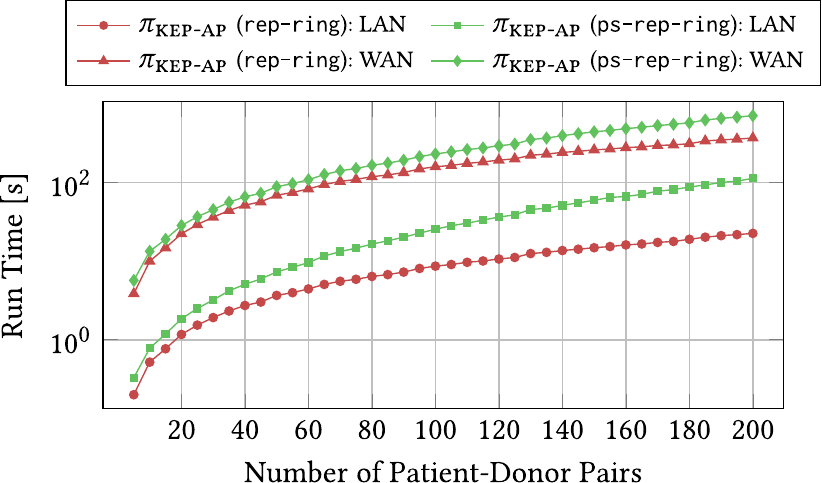}
	\caption{Run time for our novel protocol~$ \kepApProtocol $ for the special case of crossover exchanges only. In the LAN setting, the latency is 1ms and the bandwidth is 1Gbps. In the WAN setting, the latency is 20ms and the bandwidth is 1Gbps.}
	\label{plt:runtimes_two_cycles}
\end{figure}

\subsection{Evaluation of the Run Time Performance}
We evaluate the modified protocol with exactly the same setup as in Section~\ref{sub:eval_runtime}.

\paragraph{Results.}
Figure~\ref{plt:runtimes_two_cycles} shows the run time of the modified protocol~$ \kepApProtocol $, i.e., for the special case of crossover exchange. The raw data can be found in Appendix~\ref{app:results}, Table~\ref{tab:exact_results}. As expected, the run time for this special case still increases sub-exponentially with the number of patient-donor pairs whereas the overall run time is significantly lower than for the case where cycles of size $ 2 $ and $ 3 $ are considered. For example, for the SMPC primitive~\texttt{rep-ring} in the LAN setting, the run time for a maximum cycle size of $ 3 $ is on average about a factor of $ 5.2 $ larger than the run time for crossover exchanges only, increasing from a factor of $ 1.2 $ for $ 5 $ patient-donor pairs to a factor of $ 14.3 $ for $ 200 $ pairs. This is due to the fact that the size of the set $ \subsets $ is $ \sum_{i = 2}^{3} \frac{\numparties !}{(\numparties - i) ! \cdot i !} $ for a maximum cycles size of $ 3 $ but only $ \frac{\numparties}{(\numparties - 2)! \cdot 2} $ when only crossover exchanges are allowed. Thus, the size of the set $ \subsets $ increases significantly faster for a maximum cycle size of~$ 3 $.
 
\paragraph{Comparison to Related Work.}
Figure~\ref{plt:comparison_two_cycles} includes the run times for our modified protocol~$ \kepApProtocol $ for crossover exchanges only compared to the protocol from~\cite{Birka_PPKidneyExchange_2022} also for the case that only crossover exchanges are allowed. Recall that the protocol from~\cite{Birka_PPKidneyExchange_2022} was implemented in the two-party ABY framework~\cite{Demmler_ABY_2015}. Thus, a direct comparison of the specific values of the run times implies the limitations as discussed in Section~\ref{sub:eval_runtime}. Still, as for the evaluation for cycles of size $ 3 $ (cf.~Section~\ref{sub:eval_runtime}), we observe that our protocol achieves significantly better run times than the protocol from~\cite{Birka_PPKidneyExchange_2022} for increasing numbers of patient-donor pairs. This holds for $ 30 $ or more patient-donor pairs. For example, for $ 40 $~pairs, our protocol outperforms the protocol from~\cite{Birka_PPKidneyExchange_2022} by a factor of $ 2.8 $. The steeper curve for the protocol from~\cite{Birka_PPKidneyExchange_2022} indicates that this gap rapidly increases for more than $ 40 $ pairs. However, \cite{Birka_PPKidneyExchange_2022} only evaluated their protocol for up to $ 40 $ pairs.
As for the case of cycles of size~$ 3 $, note that the steeper curve of \cite{Birka_PPKidneyExchange_2022} compared to our protocol holds irrespective of the underlying framework.
 
\begin{figure}
	\includegraphics*[scale=1.0]{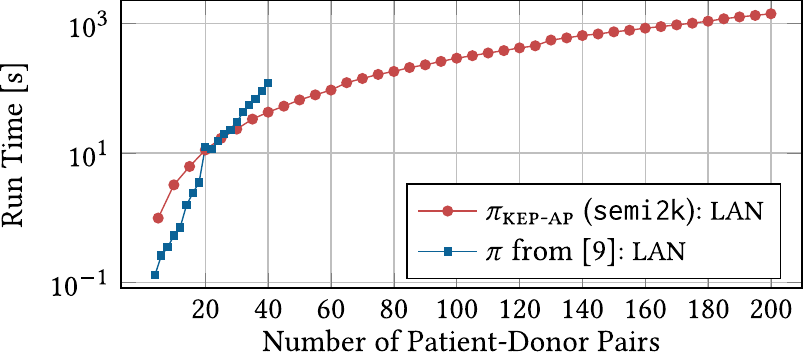}
	\caption{Run time [s] for the special case that only crossover exchanges are considered for our protocol~$ \kepApProtocol $ implemented in MP-SPDZ using the underlying SMPC primitive~\texttt{semi2k} and the results from~\cite{Birka_PPKidneyExchange_2022} (online phase only) which uses the two-party framework ABY~\cite{Demmler_ABY_2015}. Note that the latency in the LAN setting from~\cite{Birka_PPKidneyExchange_2022} is 1.3ms compared to 1ms for our LAN setting.}
	\label{plt:comparison_two_cycles}
\end{figure}

\subsection{Evaluation of the Approximation Quality}\label{app:quality_crossover}

In Section~\ref{sub:eval_quality}, we evaluated the approximation quality obtained by our protocol~$ \kepApProtocol $ with a maximum cycle size of $ 3 $ and of the protocol from~\cite{Birka_PPKidneyExchange_2022} both for the case that only cycles of size $ 2 $ are considered and for the case that only cycles of size $ 3 $ are considered. The base line against which we compared the approximation quality was an optimal solution for the KEP with a maximum cycle size of~$ 3 $, which is the common setting in practice~\cite{AshlagiKidneyExchangeOperations2021,Biro_EuropeanModellingKE_2019}. In this section, we evaluate the approximation quality of our modified protocol~$ \kepApProtocol $ and the protocol from~\cite{Birka_PPKidneyExchange_2022} for the case that only crossover exchanges are allowed. The base line is then an optimal solution for the KEP for crossover exchanges only (obtained using an LP solver for the~KEP).

\begin{figure}[t]
	\includegraphics*[scale=1.0]{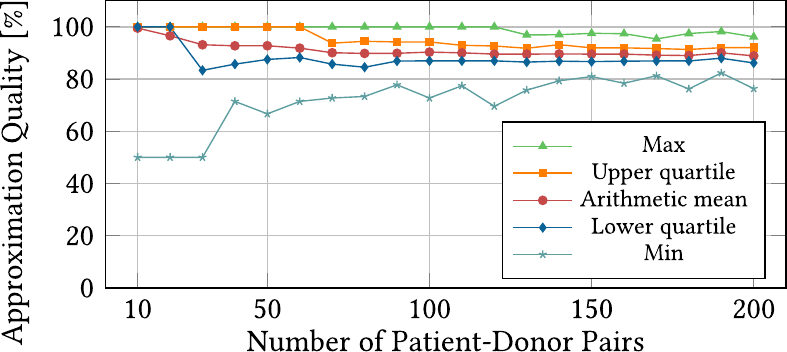}
	\caption{Quality of the solution obtained by our modified protocol~$ \kepApProtocol $ compared to an optimal solution for the KEP for the case that only crossover exchanges are considered.}
	\label{plt:offline_eval_crossover}
\end{figure}

Figure~\ref{plt:offline_eval} shows the results for the quality evaluation of our modified protocol~$ \kepApProtocol $, i.e., for the case that only crossover exchanges are allowed. As in Section~\ref{sub:eval_quality}, we choose the prioritization function that maximizes the overall number of transplants. 

We observe that the approximation quality is never below $ 50\% $ and on average around $ 89\% $ of the patients are matched if more than $ 50 $ patient-donor pairs are considered. For smaller numbers of pairs, the approximation quality is even better. For example, for $ 20 $ pairs, still more than $ 96\% $ of the patients are matched on average. As for the case of a maximum cycle size of $ 3 $, the average approximation quality is thus significantly better than the theoretical approximation ratio, which is $ 1/2 $ if only crossover exchanges are considered. This approximation ratio can be proven analogously to the proof of Theorem~\ref{thm:approx_ratio}. Compared to the results for a maximum cycle size of~$ 3 $, we also observe that the average approximation quality is better for the case of crossover exchange. Thus, our approximation protocol is even better suited for the case of crossover exchanges than for the case where cycles of size $ 2 $ and $ 3 $ are allowed. In particular, on average the approximation quality is about $ 10\% $ larger for crossover exchanges only than for a maximum cycle size of $ 3 $. This is due to the fact that there are a lot more cycles in the compatibility graph if the maximum cycle size is $ 3 $. Thus, there is also a higher probability that the approximation protocol chooses a cycle that is not part of the optimal solution.

\begin{figure}
	\includegraphics*[scale=1.0]{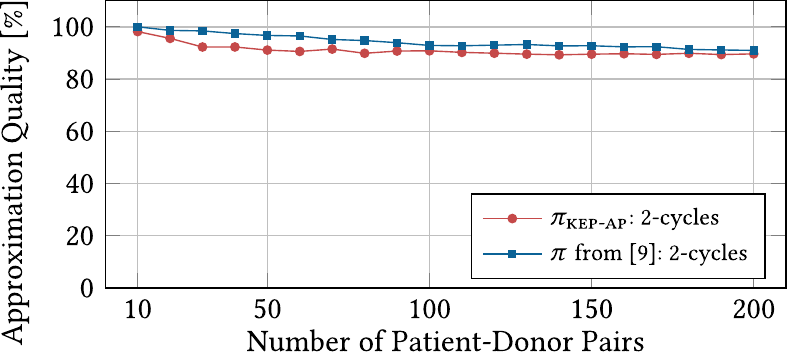}
	\caption{Arithmetic mean of the approximation quality obtained by our modified protocol $ \kepApProtocol $ compared to the protocol from~\cite{Birka_PPKidneyExchange_2022} for the special case of crossover exchanges only.}
	\label{plt:offline_eval_comparison_crossover}
\end{figure}

Figure~\ref{plt:offline_eval_comparison_crossover} shows the average approximation quality of our modified protocol compared to the protocol from~\cite{Birka_PPKidneyExchange_2022} for the case that only crossover exchanges are allowed. We observe that for this setting the average approximation quality of \cite{Birka_PPKidneyExchange_2022} is slightly better than for our protocol. In particular, the approximation quality of \cite{Birka_PPKidneyExchange_2022} is on average $ 3.4\% $ better than for our protocol. This is due to the fact that \cite{Birka_PPKidneyExchange_2022} evaluates one solution for each cycle that exists in the private compatibility graph and then chooses the best among these. In some cases, this yields a slightly better solution than our approach, which only computes one possible approximation (cf.~Algorithm~\ref{alg:greedy}). However, this comes at the cost of data obliviousness since the flow of execution of their protocol depends on the number of cycles that exist in the private compatibility graph. Furthermore, the difference in the approximation quality is very small for the large numbers of patient-donor pairs which are to be expected in practice (e.g., about $ 200 $ pairs in the case of UNOS). For example, for $ 200 $ pairs, the approximation quality of~\cite{Birka_PPKidneyExchange_2022} is only $ 1.3\% $ better than for our protocol. Besides, recall that \cite{Birka_PPKidneyExchange_2022} does not scale for such a large number of pairs in practice. Thus, the slightly superior approximation quality of the approach from~\cite{Birka_PPKidneyExchange_2022} for the special case of crossover exchange does not yield an advantage in practice over our modified protocol.

\subsection{Simulation Results}

In Section~\ref{sec:simulation}, we simulated the use of our protocol~$ \kepApProtocol $ in a dynamic kidney exchange platform where the patient-donor pairs register and de-register over time. In this section, we provide the simulation results for the special case of crossover exchanges. The simulation setup remains exactly the same as described in Section~\ref{sub:sim_setup} with the exception that both the conventional model and the $ \kepApProtocol $ model only consider crossover exchanges. Thus, we compare the number of transplants that can be achieved in the conventional model for crossover exchanges only to the number of transplants that can be achieved in the $ \kepApProtocol $ model for crossover exchanges only. In particular, we compute the percentage of transplants achieved in the $ \kepApProtocol $ model compared to those achieved in the conventional model. 

\begin{figure}[t]
	\includegraphics*[scale=0.96]{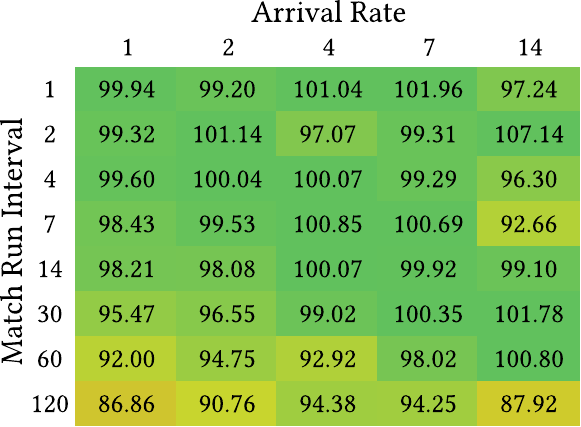}
	\caption{$ \kepApProtocol $ model vs.\ conventional model - comparison of the number of transplants measured in the $ \kepApProtocol $ model compared to the conventional model for a departure rate of $ 400 $ days and a match refusal rate of $ 20 $\%, considering the setting where only crossover exchanges are allowed. The percentages state the number of transplants achieved in the $ \kepApProtocol $ model divided by those achieved in the conventional model.}
	\label{fig:heatmap_cv_ap_crossover}	
\end{figure}

Figure~\ref{fig:heatmap_cv_ap_crossover} shows these percentages depending on the arrival rate and the match run interval. Recall that the arrival rate states the frequency at which new patient-donor pairs arrive, e.g., an arrival rate of $ 2 $ states that every $ 2 $ days a new pair arrives. The match run interval describes the frequency at which a match run is executed, e.g., a match run interval of $ 2 $ states that every $ 2 $ days a match run is executed. It still holds that the departure rate and the match refusal rate only have a very small impact on the performance difference between the two models. Therefore, we only show the results for fixed values of these parameters and rather analyze the influence of the arrival rate and match run interval, which heavily impact the performance difference between the two models. 

We essentially make the same three main observations as for the case with maximum cycle size $ \cycleBound = 3 $ (cf.~Section~\ref{sub:sim_results}).
First, in most cases the $ \kepApProtocol $ model performs better than suggested by the quality evaluation in the static setting (cf.~Section~\ref{app:quality_crossover}). 
Second, the $ \kepApProtocol $ model performs better for small match run intervals. 
Third, the $ \kepApProtocol $ model performs better for larger arrival rates.

Compared to the results for a maximum cycle size of $ 3 $, we observe that the percentages are higher for the case of crossover exchanges. In particular, the average percentage for the case of crossover exchange (cf.~Figure~\ref{fig:heatmap_cv_ap_crossover}) is about $ 7\% $ larger than for a maximum cycle size of $ 3 $ (cf.~Figure~\ref{tab:heatmap_cv_ap}). This goes in line with the results for the static setting, where our protocol~$ \kepApProtocol $ also achieves a better approximation quality for the case that only crossover exchanges are allowed.

\section{Numerical Evaluation Results}\label{app:results}

Table~\ref{tab:exact_results} presents the raw data for the run time of our protocol~$ \kepApProtocol $ underlying the plots in Figure~\ref{plt:runtimes} and Figure~\ref{plt:runtimes_two_cycles} as well as the amount of network traffic resulting from our protocol~$ \kepApProtocol $. The network traffic corresponds to the accumulated data sent by the three computing peers during the protocol execution.

\begin{table*}
	\caption{Run time [s] and network traffic [MB] for our protocol~$ \kepApProtocol $ both for the special case of crossover exchanges only and the common case where cycles up to size $ 3 $ are allowed. For each case, we consider three security models: (1) Semi-honest model with honest majority (\texttt{rep-ring}), (2) malicious model with honest majority (\texttt{ps-rep-ring}), (3) semi-honest model with dishonest majority (\texttt{semi2k}). Then, we measure the run time and network traffic both in the LAN setting with a bandwidth of 1Gbps and a latency of 1ms, and the WAN setting with a bandwidth of 1Gbps and a latency of 20ms. Note that the amount of traffic does not depend on the network setting.}
	\label{tab:exact_results}
	\begin{adjustbox}{width=\textwidth}
\centering
\begin{tabular}{c | r r r |r r r| r r | r r r| r r r| r r }
\multirow{4}{*}{\makecell{patient-donor\\pairs}} & \multicolumn{8}{c|}{crossover exchange} & \multicolumn{8}{c}{max.\ cycle size 3} \\\cline{2-17}
& \multicolumn{3}{c|}{\texttt{rep-ring}} & \multicolumn{3}{c|}{\texttt{ps-rep-ring}} & \multicolumn{2}{c|}{\texttt{semi2k}} & \multicolumn{3}{c|}{\texttt{rep-ring}} & \multicolumn{3}{c|}{\texttt{ps-rep-ring}} & \multicolumn{2}{c}{\texttt{semi2k}}\Tstrut\Bstrut\\
& \multicolumn{2}{c}{runtime} & \multirow{2}{*}{traffic}& \multicolumn{2}{c}{runtime} & \multirow{2}{*}{traffic} & \multicolumn{1}{c}{runtime} & \multirow{2}{*}{traffic} & \multicolumn{2}{c}{runtime} & \multirow{2}{*}{traffic}& \multicolumn{2}{c}{runtime} & \multirow{2}{*}{traffic}&\multicolumn{1}{c}{runtime} & \multirow{2}{*}{traffic}\Tstrut\Bstrut\\
& \small LAN & \multicolumn{1}{c}{\small WAN} &  & \multicolumn{1}{c}{\small LAN} & \multicolumn{1}{c}{\small WAN} &  & \multicolumn{1}{c}{\small LAN} & & \multicolumn{1}{c}{\small LAN} & \multicolumn{1}{c}{\small WAN} &  & \multicolumn{1}{c}{\small LAN} & \multicolumn{1}{c}{\small WAN} &  & \multicolumn{1}{c}{\small LAN} &  \Bstrut\\\hline
\Tstrut5 & 0.2 & 4 & 0.2 & 0.3 & 6 & 3 & 1 & 48 & 0.2 & 5 & 0.2 & 0.4 & 6 & 3 & 1 & 49\\
10 & 0.5 & 10 & 0.5 & 0.8 & 13 & 11 & 3 & 173 & 0.6 & 12 & 0.8 & 0.9 & 16 & 13 & 4 & 188\\
15 & 0.8 & 15 & 1 & 1 & 19 & 24 & 6 & 363 & 1 & 20 & 2 & 2 & 25 & 32 & 8 & 440\\
20 & 1 & 22 & 2 & 2 & 29 & 44 & 11 & 671 & 2 & 29 & 6 & 3 & 38 & 72 & 16 & 923\\
25 & 2 & 29 & 3 & 3 & 37 & 68 & 17 & 1024 & 2 & 37 & 12 & 4 & 49 & 133 & 28 & 1619\\
30 & 2 & 36 & 4 & 3 & 45 & 97 & 24 & 1454 & 3 & 50 & 24 & 6 & 69 & 240 & 45 & 2782\\
35 & 2 & 44 & 6 & 4 & 57 & 143 & 34 & 2147 & 3 & 57 & 40 & 8 & 85 & 390 & 68 & 4360\\
40 & 3 & 52 & 8 & 5 & 66 & 186 & 43 & 2768 & 4 & 72 & 70 & 12 & 114 & 632 & 106 & 6844\\
45 & 3 & 57 & 10 & 6 & 74 & 234 & 53 & 3461 & 5 & 80 & 104 & 16 & 136 & 918 & 149 & 9545\\
50 & 4 & 69 & 13 & 7 & 89 & 296 & 66 & 4312 & 6 & 96 & 166 & 23 & 178 & 1389 & 226 & 14294\\
55 & 4 & 75 & 16 & 8 & 98 & 359 & 79 & 5182 & 7 & 105 & 230 & 30 & 210 & 1900 & 298 & 18957\\
60 & 4 & 84 & 20 & 10 & 109 & 432 & 94 & 6167 & 8 & 122 & 345 & 41 & 267 & 2740 & 428 & 27515\\
65 & 5 & 95 & 26 & 12 & 127 & 561 & 123 & 8104 & 10 & 131 & 452 & 52 & 319 & 3619 & 556 & 35770\\
70 & 6 & 104 & 32 & 13 & 142 & 655 & 142 & 9354 & 11 & 146 & 600 & 68 & 388 & 4769 & 714 & 45996\\
75 & 6 & 110 & 37 & 15 & 152 & 753 & 165 & 10663 & 14 & 166 & 823 & 88 & 477 & 6333 & 961 & 62330\\
80 & 6 & 119 & 44 & 17 & 167 & 866 & 183 & 12118 & 16 & 184 & 1048 & 111 & 575 & 8058 & 1183 & 77590\\
85 & 7 & 126 & 51 & 18 & 178 & 980 & 210 & 13626 & 19 & 198 & 1290 & 135 & 673 & 9920 & 1435 & 93858\\
90 & 7 & 135 & 59 & 20 & 194 & 1113 & 231 & 15296 & 22 & 218 & 1610 & 167 & 804 & 12357 & 1761 & 115050\\
95 & 8 & 150 & 72 & 23 & 215 & 1272 & 261 & 17426 & 26 & 243 & 2096 & 210 & 957 & 15667 & 2286 & 150318\\
100 & 9 & 160 & 83 & 26 & 233 & 1426 & 293 & 19347 & 30 & 266 & 2542 & 253 & 1134 & 19040 & 2745 & 179748\\
105 & 9 & 167 & 93 & 28 & 248 & 1582 & 320 & 21300 & 34 & 286 & 3005 & 300 & 1325 & 22551 & 3221 & 210058\\
110 & 10 & 177 & 105 & 31 & 265 & 1759 & 353 & 23457 & 39 & 312 & 3596 & 359 & 1543 & 27016 & 3812 & 248551\\
115 & 10 & 184 & 117 & 33 & 278 & 1937 & 384 & 25637 & 45 & 335 & 4202 & 419 & 1787 & 31612 & 4482 & 287883\\
120 & 11 & 194 & 132 & 37 & 296 & 2136 & 423 & 28039 & 54 & 379 & 5345 & 515 & 2117 & 39233 & 5691 & 370346\\
125 & 11 & 202 & 145 & 39 & 311 & 2337 & 456 & 30456 & 60 & 406 & 6136 & 589 & 2411 & 45184 & 6574 & 421522\\
130 & 13 & 224 & 176 & 46 & 355 & 2818 & 559 & 37637 & 69 & 442 & 7126 & 692 & 2796 & 52796 & 7575 & 488841\\
135 & 13 & 232 & 192 & 47 & 371 & 3050 & 604 & 40462 & 78 & 476 & 8119 & 791 & 3193 & 60278 & 8544 & 552762\\
140 & 14 & 244 & 212 & 52 & 399 & 3313 & 657 & 43570 & 88 & 519 & 9349 & 909 & 3667 & 69514 & 9877 & 631799\\
145 & 14 & 252 & 231 & 56 & 422 & 3572 & 694 & 46649 & 97 & 554 & 10576 & 1027 & 4119 & 78765 & 11020 & 710367\\
150 & 15 & 262 & 253 & 60 & 446 & 3862 & 749 & 50034 & 117 & 630 & 13030 & 1223 & 4767 & 94705 & 13697 & 889389\\
155 & 15 & 269 & 274 & 65 & 464 & 4146 & 791 & 53376 & 128 & 675 & 14552 & 1367 & 5334 & 106107 & 15251 & 986976\\
160 & 16 & 280 & 300 & 67 & 490 & 4468 & 845 & 57049 & 144 & 733 & 16417 & 1555 & 6045 & 120020 & 17065 & 1106490\\
165 & 17 & 287 & 324 & 71 & 509 & 4779 & 900 & 60663 & 157 & 783 & 18246 & 1728 & 6747 & 133736 & 19057 & 1223220\\
170 & 17 & 298 & 352 & 78 & 535 & 5134 & 959 & 64634 & 177 & 853 & 20472 & 1943 & 7574 & 150367 & 21192 & 1365220\\
175 & 18 & 304 & 379 & 82 & 556 & 5470 & 1017 & 68528 & 196 & 915 & 22650 & 2159 & 8404 & 166702 & 23221 & 1503590\\
180 & 19 & 314 & 410 & 88 & 580 & 5854 & 1095 & 72806 & 216 & 996 & 25283 & 2465 & 9398 & 186397 & 25791 & 1670900\\
185 & 20 & 341 & 476 & 94 & 630 & 6408 & 1192 & 80187 & 245 & 1092 & 30168 & 2734 & 10496 & 217021 & 31191 & 2036290\\
190 & 21 & 352 & 513 & 101 & 661 & 6833 & 1274 & 84886 & 272 & 1190 & 33333 & 3045 & 11681 & 240513 & 34405 & 2238530\\
195 & 22 & 360 & 546 & 105 & 684 & 7236 & 1344 & 89440 & 304 & 1274 & 36386 & 3387 & 12795 & 263300 & 37414 & 2432620\\
200 & 23 & 372 & 586 & 114 & 715 & 7695 & 1426 & 94467 & 325 & 1378 & 40057 & 3741 & 14176 & 290609 & 41081 & 2666330\\
\end{tabular}
\end{adjustbox}
\end{table*}

\end{document}